\documentclass[a4paper,11pt]{article}
\pdfoutput=1 
\usepackage{jheppub} 
\usepackage[T1]{fontenc}

\usepackage[english]{babel}
\usepackage[utf8]{inputenc}
\usepackage{mdwlist}
\usepackage{verbatim}
\usepackage{accents}
\usepackage{subcaption}
\usepackage{placeins}
\usepackage{amsmath,amssymb}

\usepackage{tikz}
\usetikzlibrary{calc}
\usetikzlibrary{arrows}
\usetikzlibrary{decorations.pathreplacing,decorations.markings}

\newtheorem{theorem}{Theorem}

\newtheorem{definition}[theorem]{Definition}

\newenvironment{proof}[1][Proof]{\noindent\textbf{#1.} }{\ \rule{0.5em}{0.5em}}

\newcommand{\bra}[1]{\langle#1|} \newcommand{\ket}[1]{|#1\rangle}
\newcommand{\ketbra}[2]{|#1\rangle\!\langle#2|}

\newcommand{\tr}{\text{tr}}
\newcommand{\rank}{\text{rank}}

\tikzset{
    >=stealth',
    punkt/.style={
           rectangle,
           rounded corners,
           draw=black, very thick,
           text width=6.5em,
           minimum height=2em,
           text centered},
    pil/.style={
           ->,
           thick,
           shorten <=2pt,
           shorten >=2pt,},
  on each segment/.style={
    decorate,
    decoration={
      show path construction,
      moveto code={},
      lineto code={
        \path [#1]
        (\tikzinputsegmentfirst) -- (\tikzinputsegmentlast);
      },
      curveto code={
        \path [#1] (\tikzinputsegmentfirst)
        .. controls
        (\tikzinputsegmentsupporta) and (\tikzinputsegmentsupportb)
        ..
        (\tikzinputsegmentlast);
      },
      closepath code={
        \path [#1]
        (\tikzinputsegmentfirst) -- (\tikzinputsegmentlast);
      },
    },
  },
  mid arrow/.style={postaction={decorate,decoration={
        markings,
        mark=at position .5 with {\arrow[#1]{stealth'}}
      }}}
}

\newcommand\bellbra[2]{
    \draw[fill] (#1,#2) circle [radius=0.1];
    \draw [thick] (#1-2,#2) -- (#1+2,#2);
    \draw[thick,->] (#1,#2) -- (#1-1,#2);
    \draw[thick,->] (#1,#2) -- (#1+1,#2);
}

\newcommand\bellket[2]{
    \draw[fill] (#1,#2) circle [radius=0.1];
    \draw [thick] (#1-2,#2) -- (#1+2,#2);
    \draw[thick,->] (#1-2,#2) -- (#1-1,#2);
    \draw[thick,->] (#1+2,#2) -- (#1+1,#2);
}

\newcommand\arrowedcurve[6]{
	\draw[thick,postaction={on each segment={mid arrow}}] (#1,#2) to [out=#5,in=#6] (#3,#4);
}

\newcommand\bellbrasegment[5]{
\arrowedsegment{#1/2+#3/2}{#2/2+#4/2}{#1}{#2}
\arrowedsegment{#1/2+#3/2}{#2/2+#4/2}{#3}{#4}
\draw[fill] (#1/2+#3/2,#2/2+#4/2) circle [radius=#5];
}

\newcommand\bluearrowedsegment[4]{
	\draw[blue,thick,postaction={on each segment={mid arrow}}] (#1,#2) -- (#3,#4);
}

\newcommand\blueketsegment[5]{
\bluearrowedsegment{#1}{#2}{#1/2+#3/2}{#2/2+#4/2}
\bluearrowedsegment{#3}{#4}{#1/2+#3/2}{#2/2+#4/2}
\draw[fill,blue] (#1/2+#3/2,#2/2+#4/2) circle [radius=#5];
}

\newcommand\bluebrasegment[5]{
\bluearrowedsegment{#1/2+#3/2}{#2/2+#4/2}{#1}{#2}
\bluearrowedsegment{#1/2+#3/2}{#2/2+#4/2}{#3}{#4}
\draw[fill,blue] (#1/2+#3/2,#2/2+#4/2) circle [radius=#5];
}

\newcommand\arrowedsegment[4]{
	\draw[thick,postaction={on each segment={mid arrow}}] (#1,#2) -- (#3,#4);
}

\newcommand\polararrowedsegment[4]{
	\draw[thick,postaction={on each segment={mid arrow}}] (#1,#2) -- +(#3:#4);
}

\title{Tensor networks for dynamic spacetimes}

\author[a]{A. May}

\affiliation[a]{Department of Physics and Astronomy, University of British Columbia, 6224 Agricultural Road, Vancouver, B.C., V6T 1W9\\ Canada}

\emailAdd{may@phas.ubc.ca}

\abstract{Existing tensor network models of holography are limited to representing the geometry of constant time slices of static spacetimes. We study the possibility of describing the geometry of a dynamic spacetime using tensor networks. We find it is necessary to give a new definition of length in the network, and propose a definition based on the mutual information. We show that by associating a set of networks with a single quantum state and making use of the mutual information based definition of length, a network analogue of the maximin formula can be used to calculate the entropy of boundary regions.}

\begin{document} 
\maketitle
\flushbottom

\section{Introduction}

In the context of the AdS/CFT correspondence, the Ryu-Takayanagi (RT) formula \cite{ryu2006holographic} relates the entropy of boundary regions to minimal surfaces of the asymptotically AdS bulk region. Building on an understanding of the MERA network \cite{vidal2007entanglement, vidal2008class} and the presence of area laws for entanglement entropy in tensor network states \cite{eisert2010colloquium}, tensor networks which realize the RT formula have been found \cite{pastawski2015holographic, hayden2016holographic, yang2016bidirectional}. These tensor network models also realize the subregion duality property \cite{czech2012gravity, bousso2012light, hubeny2012causal, bousso2013null, almheiri2015bulk}, a fact that has highlighted the role of quantum error correction in AdS/CFT.

In the analogy between tensor network models and AdS/CFT the graph geometry of the network corresponds to the geometry of a spacelike slice of an AdS spacetime, while the state on the free legs of the network corresponds to the CFT state. Although the analogy realizes notable properties of the AdS/CFT correspondence, we point out here that the bulk geometries represented by these models are restricted to constant time slices of static spacetimes. 

In this paper we study if it is possible to use tensor networks to describe geometries and boundary states in the dynamic case. To evaluate if our networks describe a dynamic geometry we study the  HRT formula \cite{hubeny2007covariant}, which generalizes RT to dynamic CFT states and AdS geometries. In the tensor network setting we find it useful to consider the restatement of HRT as the maximin formula \cite{wall2014maximin}. The maximin formula involves optimizing over all the spacelike slices anchored on a boundary region of interest. In particular, for a dynamic spacetime minimal surfaces for different boundary regions may not lie in a single spacelike slice.

In the tensor network picture there are many networks which contract to give the same boundary state. In our construction these networks play the role of the many possible spacelike slices of the AdS interior for a fixed boundary region. In particular, we show that the set of networks all contracting to a fixed boundary state can be searched over to calculate boundary entropies. In the case of HaPPY networks we find that this procedure reduces to finding minimal lengths in a single network. Additionally, we demonstrate the existence of states where several networks are needed to calculate the entropies of all boundary regions, analogous to the situation for dynamic CFT states.

A key step in constructing network models that describe dynamic geometries is a redefinition of length. Indeed, in earlier models length is defined as the logarithm of the total dimension of the legs crossed by the cut of interest, that is $L(\gamma) = \log \dim \gamma$. It is not difficult to see that this definition of length removes any possibility of describing a dynamic geometry. Indeed, in the maximin picture there should be networks where the minimal length of a cut drops below the entropy of the boundary region it encloses. For instance taking the maximin surface and moving it in a timelike direction decreases its length. However, the quantity $\log \dim \gamma$ actually provides an upper bound on the entropy of the enclosed region, so if $L(\gamma) = \log \dim \gamma$ it can never happen that the length is less than the entropy. 

We emphasize that our results do not contain a method for evolving a tensor network in time. Rather, we establish a network analogue of the maximin formula. Our procedure searches over a set of networks, but this set does not have an ordering and consequently we cannot directly interpret the set as a sequence of networks describing a time evolution. However, one motivation for this work is that it may give insight into how to construct a time evolving network. We discuss this possibility further in section \ref{sec:discussion}.

The structure of this paper is as follows. In section \ref{sec:background} we give background on tensor networks. The reader familiar with networks may skim this section, but should note the explanation of the arrow notation and the statement of theorem \ref{thm:converse}, which is straightforward but not stated elsewhere. Section \ref{sec:stateoncut} discusses tensor networks as maps and points out the equivalence of isometries and minimal cuts in HaPPY networks. Section \ref{sec:maximinlessons} gives our demonstration that existing constructions are limited to representing constant time slices of static spacetimes. Section \ref{sec:lengthdefined} and \ref{sec:staticexample} give a new definition of length. Finally, in section \ref{sec:dynamicexample} we give a preliminary example of a set of networks and a quantum mechanical state with properties analogous to dynamic holographic states. 

\section{Tensor network background} \label{sec:background}

\subsection{Tensor network basics} \label{sec:basics}

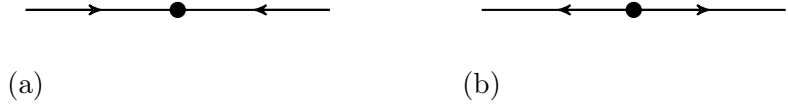
\begin{figure} 
\begin{tikzpicture}
\bellket{0}{0}
\node at (-2,-1) {(a)};

\bellbra{6}{0}
\node at (4,-1) {(b)};
\node at (0,3) {};
\end{tikzpicture}
\caption{(a) A maximally entangled state in the Hilbert space, represented in the graphical notation. (b) A maximally entangled state in the dual Hilbert space, represented in the graphical notation.}
\label{fig:bellpairs}
\end{figure}

The tensor network formalism describes graphically the pattern of contraction of a set of simple objects to form a more complex quantum state. The basic objects in the graphical formalism are vertices with some number of lines attached. Each vertex corresponds to a quantum state, and the lines each correspond to a ket or bra index. For instance
\begin{align}
\ket{\Psi^+} = \sum_{m=1}^D \ket{m}\ket{m}
\end{align}
is represented by figure \ref{fig:bellpairs}a. We attach a direction (inward or outward) to each line in the diagram, with inward arrows indicating ket indices and outward arrows indicating bra indices. Thus $\bra{\Psi^+}$ is represented as in figure \ref{fig:bellpairs}b. 

It will be convenient to write down quantum states without explicitly including their basis vectors. For example, the maximally entangled state $\ket{\Psi^+}$ is written $\delta_{ab}$, leaving the choice of basis implicit. In this notation ket indices are lowered and bra indices are raised, so $\bra{\Psi^+}$ becomes $\delta^{ab}$. More generally a quantum state
\begin{align} \label{eq:examplestate}
\ket{\phi} = \sum_{a,b} T_{ab} \ket{a}\ket{b}
\end{align}
is specified as $T_{ab}$. The upper and lower indices carry transformation rules with them. Since
\begin{align}
\ket{\phi} &= \sum_{a,b,c} T_a (U^\dagger)_{b}^a U^b_c \ket{c} = \sum_{b} T_b^\prime  \ket{b},
\end{align}
we see that lower indices transform according to $T_a\rightarrow T_a (U^\dagger)_{b}^a$ under a change of basis described by $\ket{b}=U^{b}_c\ket{c}$. Similarly under the same change of basis upper indices transform according to $T^a\rightarrow T^a U_a^b$.

\begin{figure}
\begin{tikzpicture}[scale=0.3]
\draw[thick,postaction={on each segment={mid arrow}}] (8,3) -- (5,0);
\draw[thick,postaction={on each segment={mid arrow}}] (2,3) -- (5,0);
\draw[thick,postaction={on each segment={mid arrow}}] (8,-3) -- (5,0);
\draw[thick,postaction={on each segment={mid arrow}}] (2,-3) -- (5,0);
\draw[fill=white] (5,0) circle [radius=1]; 
\node at (5,0) {S};

\draw[thick,postaction={on each segment={mid arrow}}] (-8,3) -- (-5,0);
\draw[thick,postaction={on each segment={mid arrow}}] (-2,3) -- (-5,0);
\draw[thick,postaction={on each segment={mid arrow}}] (-8,-3) -- (-5,0);
\draw[thick,postaction={on each segment={mid arrow}}] (-2,-3) -- (-5,0);
\draw[fill=white] (-5,0) circle [radius=1];
\node at (-5,0) {T};

\draw[->,thick] (10,0)--(12,0);

\draw[fill] (21,3) circle [radius=0.25];
\draw[fill] (21,-3) circle [radius=0.25];

\draw[thick,postaction={on each segment={mid arrow}}] (21,3) to [out=0,in=135] (26,0);
\arrowedcurve{21}{3}{26}{0}{0}{135}
\arrowedcurve{21}{3}{16}{0}{180}{45}
\arrowedcurve{21}{-3}{16}{0}{180}{-45}
\arrowedcurve{21}{-3}{26}{0}{0}{-135}

\draw[thick,postaction={on each segment={mid arrow}}] (29,3) -- (26,0);
\draw[thick,postaction={on each segment={mid arrow}}] (29,-3) -- (26,0);
\draw[fill=white] (5+21,0) circle [radius=1]; 
\node at (5+21,0) {S};

\draw[thick,postaction={on each segment={mid arrow}}] (13,3) -- (16,0);
\draw[thick,postaction={on each segment={mid arrow}}] (13,-3) -- (16,0);
\draw[fill=white] (-5+21,0) circle [radius=1];
\node at (-5+21,0) {T};

\end{tikzpicture}
\caption{A basic example of a contraction of two quantum states into a tensor network. Algebraically, the objects at left are written $T_{abcd}$ and $S_{efgh}$. The object at right is $T_{abcd}\delta^{ce}\delta^{df}S_{efgh}$.}
\label{fig:examplenetwork}
\end{figure}
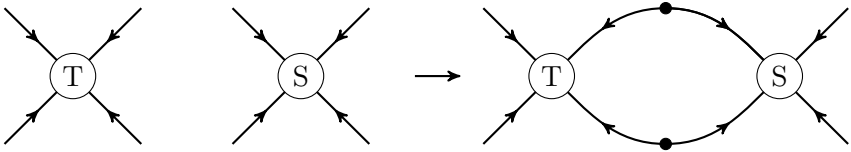

The basic operation of the tensor network formalism is the composition of two quantum states. Composition of two ket states, say $T_{abcd}$ and $S_{efgh}$ is performed by introducing maximally entangled bra states, 
\begin{align}
T_{abcd} \circ S_{efgh} \rightarrow T_{abcd}\delta^{ce}\delta^{df} S_{efgh}.
\end{align}
This is illustrated in figure \ref{fig:examplenetwork}. The contraction performed above was not the unique choice, since different pairs of indices could have been contracted. In general, to describe a pattern of contraction of two or more quantum states a graph is specified. States are associated with vertices, with each line attached to a vertex representing a particular index. The contraction is performed by placing maximally entangled pairs on the edges and connecting in going and out going lines. 

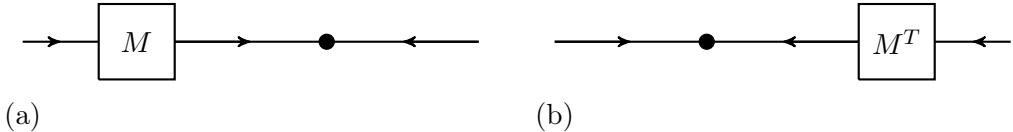
\begin{figure} 
\begin{tikzpicture}

\draw[thick,fill=white] (-0.5,-0.5)--(+0.5,-0.5)--(+0.5,+0.5)--(-0.5,0.5)--(-0.5,-0.5);
\bellket{2.5}{0}
\node at (0,0) {$M$};
\draw[thick,postaction={on each segment={mid arrow}}] (-1.5,0) -- (-0.5,0);
\node at (-1.5,-1) {(a)};

\bellket{7.5}{0}
\draw[thick,fill=white] (9.5,-0.5)--(10.5,-0.5)--(10.5,+0.5)--(9.5,+0.5)--(9.5,-0.5);
\node at (10,0) {$M^T$};
\draw[thick,postaction={on each segment={mid arrow}}] (11.5,0) -- (10.5,0);
\node at (5.5,-1) {(b)};

\end{tikzpicture}
\caption{(a) An operator $M\otimes \mathcal{I}$ applied to a quantum state. (b) In the case where the vertex represents a maximally entangled state, the operator $M$ can be moved to the other subspace by taking the transpose.}
\label{fig:operatordiagram}
\centering
\end{figure}

Operators acting on quantum states we represent as vertices having both inward directed and outward directed lines attached, and are written as tensors with upper and lower indices, for example ${M_a}^b$. Diagrammatically, applying an operator to a state is given by connecting lines. The algebraic equivalent is performing the appropriate sum. Thus the operator ${M_{a}}^b$ applied to a state $T_{ab}$ is represented by figure \ref{fig:operatordiagram}a or by ${M_a}^c T_{cb}$. In the case of maximally entangled states, it is straightforward to show the identity
\begin{align}
M \otimes \mathcal{I} \ket{\Psi^+} = \mathcal{I} \otimes M^T \ket{\Psi^+}.
\end{align}
We will refer to this as the \emph{transpose rule} below.

Indices can be raised and lowered by contracting with maximally entangled pairs. In particular an operator ${M_{a}}^b$ can be mapped to a state by ${M_a}^b \rightarrow M_{ab}={M_{a}}^c \delta_{cb}$, and states to operators by $M_{ab} \rightarrow {M_a}^b = M_{ac}\delta^{cb}$. In the simplest case of a state with two indices this is also known as the Choi-Jamio{\l}kowski mapping \cite{jiang2013channel} between pure bipartite states and operators. More generally, we can raise and lower indices on objects with arbitrary numbers of indices by appropriate contractions with maximally entangled pairs. 

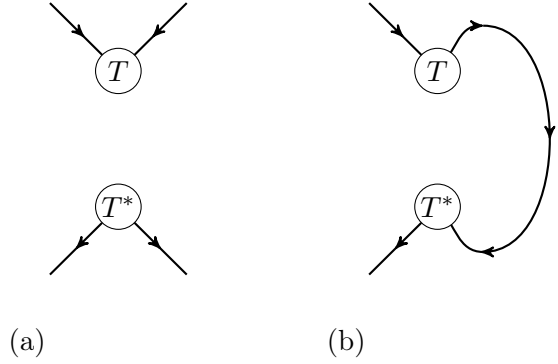
\begin{figure} 
\centering
\begin{tikzpicture}[scale=0.3]

\node at (0,0) {(a)};
\node at (14,0) {(b)};

\arrowedsegment{1}{15}{4}{12}
\arrowedsegment{7}{15}{4}{12}
\arrowedsegment{3.29}{5.29}{1}{3}
\arrowedsegment{4.71}{5.29}{7}{3}

\arrowedsegment{15}{15}{18}{12}
\arrowedsegment{17.29}{5.29}{15}{3}
\draw[thick,->] (18,12) to [out=45,in=180] (20,14);
\arrowedcurve{20}{14}{20}{4}{0}{0}
\draw[thick,->] (20,4) -- (19.9,4);
\draw[thick,->] (20,4) to [out=180,in=-45] (18,6);

\draw[fill=white] (4,6) circle [radius=1];
\draw[fill=white] (4,12) circle [radius=1];
\node at (4,6) {$T^*$};
\node at (4,12) {$T$};

\draw[fill=white] (18,6) circle [radius=1];
\draw[fill=white] (18,12) circle [radius=1];
\node at (18,6) {$T^*$};
\node at (18,12) {$T$};

\end{tikzpicture}
\caption{(a) Representation of the density matrix $\rho_{AB}$ shown in eq. \ref{eq:densitymatrixexample} in the graphical notation. (b) Representation of the reduced density matrix $\rho_A$ in the graphical notation.}
\label{fig:densitymatrix}
\centering
\end{figure}

It is also possible to represent density matrices in a tensor network diagram. The density matrix corresponding to the state in eq. \ref{eq:examplestate}, given by 
\begin{align} \label{eq:densitymatrixexample}
\rho_{AB}=\sum_{ab}T_{ab} (T^*)^{cd} \ketbra{a}{c}_A\otimes\ketbra{b}{d}_B,
\end{align}
is drawn as the network shown in figure \ref{fig:densitymatrix}a. Contracting corresponding inward or outward indices in the diagram performs the partial trace. We show the diagram for $\rho_A$ in figure \ref{fig:densitymatrix}b.

In general the contraction of two properly normalized quantum states results in an unnormalized output, meaning it is necessary to add a final normalization factor after all the contractions have been performed. For this reason we frequently drop any normalization factors on our initial states, for instance writing $\ket{\Psi^+} = \sum \ket{m}\ket{m}$, since this has no effect on the final state after contraction and adding proper normalization. 

We recall an important bound on the von Neumann entropy of a subsystem of a tensor network state. Suppose we have a quantum state which is written
\begin{align}
\ket{\psi} &= \sum T_{i_1...i_n j_1...j_n}\ket{i_1}_{\bar{A}}...\ket{i_n}_{\bar{A}}\ket{j_1}_A... \ket{j_n}_A \nonumber \\
&= \sum_{IJ}T_{IJ} \ket{I}_{\bar{A}} \ket{J}_A,
\end{align}
where the capital indices stand in for a set of lower case indices, and we are interested in the entropy of the $A$ subsystem. If this state is described by a tensor network we can consider a cut $\gamma$ passing through the network and separating off the region $A$. Such a cut is specified by a path in the dual graph, which passes through a sequence of maximally entangled pairs. For each cut, there is a corresponding decomposition of the $T_{IJ}$ given by
\begin{align}
\ket{\psi} &= \sum_{IJKL}A_{IJ}\delta^{JK}{B}_{KL} \ket{I}_{\bar{A}} \ket{L}_A.
\end{align}
Now define states by $\sum_J B_{KL} \ket{L}_A = \ket{\hat{K}}_A$ and $\sum_I A_{IK} \ket{I} = \ket{\hat{K}}_{\bar{A}}$. This gives
\begin{align}
\ket{\psi} &= \sum_{K=1}^{|K|} \ket{\hat{K}}_{\bar{A}} \ket{\hat{K}}_A.
\end{align}
From this we have that $\rank(\rho_A) \leq |K|$. Since the von Neumann entropy is bounded above by the log of the rank, we have
\begin{align} \label{eq:rankbound}
S(\rho_A) \leq \log{\dim \gamma},
\end{align}
where we define the dimension of the cut by $\dim \gamma \equiv |K|$. Equality occurs when $\ket{\hat{K}}_A$ and $\ket{\hat{K}}_{\bar{A}}$ are orthonormal bases. 

HaPPY networks are a special class of networks for which it has been shown \cite{pastawski2015holographic} that cuts which cross a minimal number of legs saturate the bound \ref{eq:rankbound}. The basic building block of a HaPPY network is a perfect tensor. We remind the reader of the definition of a perfect tensor below. 

\begin{definition}
A \textbf{perfect tensor} is a tensor $T_{a_1 a_2...a_{2n}}$ with an even number of indices and having the property that
\begin{align} \label{eq:perfectioncondition}
T_{a_1...a_n a_{n+1}...a_{2n}} (T^*)^{b_1...b_n a_{n+1}...a_{2n}} = \delta_{a_1}^{b_1}...\delta_{a_n}^{b_n},
\end{align}
where the $a_{n+1}...a_{2n}$ can be chosen to be any of the $2n$ legs of the tensor.
\end{definition}

We can also raise and lower legs on the left side of \ref{eq:perfectioncondition}, giving
\begin{align}\label{eq:unitarycondition}
{T_{a_1...a_n}}^{a_{n+1}...a_{2n}} {(T^*)^{b_1...b_n}}_{a_{n+1}...a_{2n}} = \delta_{a_1}^{b_1}...\delta_{a_n}^{b_n}.
\end{align}
This shows we can think of the perfection condition as the statement that the tensor defines a unitary transformation from any set of $n$ legs to the complement. It follows by contracting indices on both sides of \ref{eq:unitarycondition} that perfect tensors define isometries from any subset of legs of size $k<n$ to the complement. We illustrate the perfection condition in figure \ref{fig:perfectcondition}. 

A useful operation involving perfect tensors is operator pushing. Suppose we have a perfect tensor $T$ and an operator $\mathcal{O}$ which acts on three legs. Then we can rewrite the tensor $\mathcal{O}T$ as $T \mathcal{O}'$ by defining $\mathcal{O}'=T^\dagger \mathcal{O} T$. We illustrate this in figure \ref{fig:operatorpushing}. An operator acting on a single leg of a $2n$ leg perfect tensor can be pushed through to any $n$ legs, but in general the operator $\mathcal{O}'$ will not act as a tensor product across those legs.  

To construct a HaPPY network, perfect tensors are placed on the vertices of a graph with a non-positive curvature condition\footnote{By non-positive curvature it is meant that distance (measured in number of legs cut) between points in the dual graph has no maximum away from the boundary.}. Reference \cite{pastawski2015holographic} which introduced HaPPY networks does not keep track of the distinction between upper and lower indices in their construction, so to translate their construction to the language used here we must consider a maximally entangled state being placed along every edge of this non-positively curved graph. This done, we may perform the contraction, leaving a boundary state whose entanglement entropies saturate \ref{eq:rankbound}.

\begin{figure}
\begin{tikzpicture}[scale=0.7]
\draw[thick, postaction={on each segment={mid arrow}}] (0,0) to [out=0,in=180] (5,0);
\draw[thick, postaction={on each segment={mid arrow}}] (0,0) to [out=45,in=135] (5,0);
\draw[thick, postaction={on each segment={mid arrow}}] (0,0) to [out=-45,in=-135] (5,0);
\draw[thick, postaction={on each segment={mid arrow}}] (-1,0)  to [out=0,in=180] (-3,0);
\draw[thick, postaction={on each segment={mid arrow}}] (-0.71,0.71)  to [out=-45,in=135] (-2,2);
\draw[thick, postaction={on each segment={mid arrow}}] (-0.71,-0.71) to [out=45,in=-135] (-2,-2);
\draw[thick, postaction={on each segment={mid arrow}}] (8,0) to [out=0,in=180] (6,0);
\draw[thick, postaction={on each segment={mid arrow}}] (7,2) to [out=45,in=-135] (5+0.71,0.71);
\draw[thick, postaction={on each segment={mid arrow}}] (7,-2) to [out=-45,in=135] (5+0.71,-0.71);
\draw[fill=white] (0,0) circle [radius=1];
\node at (0,0) {$T^\dagger$};
\draw[fill=white] (5,0) circle [radius=1];
\node at (5,0) {$T$};

\node at (10,0) {$=$};
\draw[thick, postaction={on each segment={mid arrow}}] (17,0)--(12,0);
\draw[thick, postaction={on each segment={mid arrow}}] (17,2)--(12,2);
\draw[thick, postaction={on each segment={mid arrow}}] (17,-2)--(12,-2);
\end{tikzpicture}
\caption{Illustration of the defining condition for perfect tensors. The same equality must hold when any subset consisting of half the legs is contracted.}
\label{fig:perfectcondition}
\end{figure}
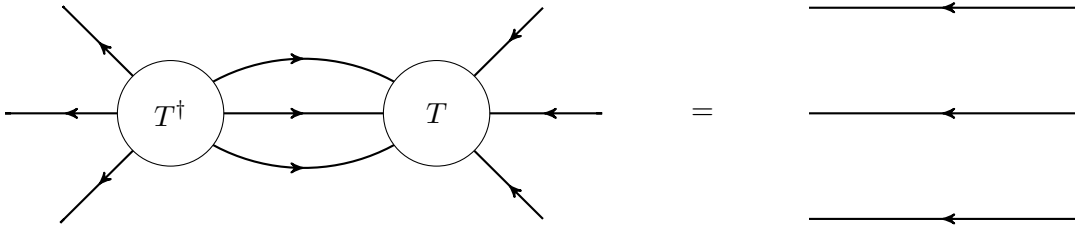

In the case of HaPPY networks and in random tensor networks the length of curves through the graph is defined by:
\begin{align}
L_G(\gamma) = \log (\dim \gamma).
\end{align}
Herein we will refer to this as the graph length. In section \ref{sec:lengthdefined} we will discuss an alternative notion of length in the network.

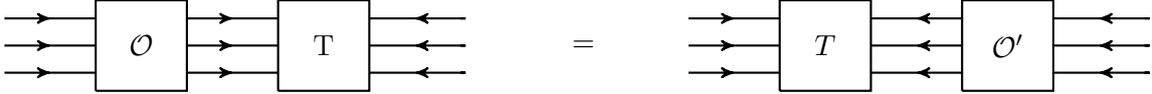
\begin{figure} 
\begin{tikzpicture}[scale=1.2]

\arrowedcurve{0.5}{0}{1.5}{0}{0}{180}
\arrowedcurve{0.5}{0.3}{1.5}{0.3}{0}{180}
\arrowedcurve{0.5}{-0.3}{1.5}{-0.3}{0}{180}
\arrowedcurve{-1.5}{0}{-0.5}{0}{0}{180}
\arrowedcurve{-1.5}{-0.3}{-0.5}{-0.3}{0}{180}
\arrowedcurve{-1.5}{0.3}{-0.5}{0.3}{0}{180}
\arrowedcurve{3.5}{0}{2.5}{0}{0}{180}
\arrowedcurve{3.5}{-0.3}{2.5}{-0.3}{0}{180}
\arrowedcurve{3.5}{0.3}{2.5}{0.3}{0}{180}

\node at (4.85,0) {$=$};

\arrowedcurve{6}{0}{7}{0}{0}{180}
\arrowedcurve{6}{0.3}{7}{0.3}{0}{180}
\arrowedcurve{6}{-0.3}{7}{-0.3}{0}{180}
\arrowedcurve{9}{0}{8}{0}{0}{180}
\arrowedcurve{9}{-0.3}{8}{-0.3}{0}{180}
\arrowedcurve{9}{0.3}{8}{0.3}{0}{180}
\arrowedcurve{11}{0}{10}{0}{0}{180}
\arrowedcurve{11}{-0.3}{10}{-0.3}{0}{180}
\arrowedcurve{11}{0.3}{10}{0.3}{0}{180}

\draw[thick,fill=white] (-0.5,-0.5)--(+0.5,-0.5)--(+0.5,+0.5)--(-0.5,+0.5)--(-0.5,-0.5);
\draw[thick,fill=white] (2-0.5,-0.5)--(2+0.5,-0.5)--(2+0.5,0.5)--(2-0.5,0.5)--(2-0.5,-0.5);
\node at (0,0) {$\mathcal{O}$};
\node at (2,0) {T};

\draw[thick,fill=white] (7,-0.5)--(8,-0.5)--(8,0.5)--(7,0.5)--(7,-0.5);
\draw[thick,fill=white] (9.5-0.5,-0.5)--(9.5+0.5,-0.5)--(9.5+0.5,+0.5)--(9.5-0.5,+0.5)--(9.5-0.5,-0.5);
\node at (9.5,0) {$\mathcal{O}'$};
\node at (7.5,0) {$T$};

\end{tikzpicture}
\caption{Illustration of the operator pushing operation. An operator $\mathcal{O}$ acting on a subset of size $n$ of a perfect tensor with $2n$ legs is equivalent to an operator $\mathcal{O}'=T^\dagger \mathcal{O} T$ acting on the complement.}
\label{fig:operatorpushing}
\end{figure}

\subsection{Maps defined from tensor networks} \label{sec:stateoncut}

We will see in section \ref{sec:graphlengthanddynamics} that describing the geometry of a spacelike slice of an evolving geometry requires a new definition of length in the network. Our definition will be built on a set of natural maps defined from the network. Earlier literature \cite{qi2013exact, pastawski2015holographic, hayden2016holographic, yang2016bidirectional} contains two types of mapping: maps between boundary legs and another set of ``bulk'' uncontracted legs, and maps between interior contracted legs and boundary legs. Our maps are of the second type. In the interest of being self contained and to fix language we describe this mapping in some detail below. 

Any cut which partitions the network defines two tensors, call them $C$ and $D$, which contract to give the boundary state. That is we can write
\begin{align}
\ket{\psi}_{A\bar{A}} = \sum_{IJK} C_{IJ}\delta^{JK}D_{KL}\ket{I_A}\ket{L_{\bar{A}}}.
\end{align}
This state can be formed by acting with the operators\footnote{If the boundary state is thought of as formed by contracting two states $\ket{C}=\sum C_{IJ}\ket{I}\ket{J}$ and $\ket{D}=\sum D_{KJ}\ket{K}\ket{J}$, these are just the operators that $\ket{C}$ and $\ket{D}$ are brought to under the Choi-Jamio{\l}kowski mapping.}
\begin{align}
C = \sum_{IJ} C_{Ij_1...j_n} \ketbra{I_A}{j_{{B}_1}} ...\bra{j_{{B}_n}}, \nonumber \\
D = \sum_{KJ} D_{K j_1...j_n} \ketbra{K_{\bar{A}}}{j_{\bar{B}_1}} ...\bra{j_{\bar{B}_n}},
\end{align}
on a collection of maximally entangled pairs. That is
\begin{align} \label{eq:slickpsi}
\ket{\psi}_{A\bar{A}} = (C\otimes D) \bigotimes_{i=1}^n \ket{\Psi^+}_{\bar{B}_i \bar{B}_i}.
\end{align}

In this picture $C$ and $D$ act as maps from an interior Hilbert space onto the boundary. There is freedom in how we choose the operators $C$ and $D$. For example, we could form the same state by contraction with a different choice of entangled states $\ket{\Psi_i}$ by writing
\begin{align} \label{eq:slickpsifree}
\ket{\psi}_{A\bar{A}} &= (C \Lambda^{-1} \otimes D \bar{\Lambda}^{-1}) \bigotimes_{i=1}^n(\lambda^i\otimes\bar{\lambda}^i)\ket{\Psi^+}_{B_i \bar{B}_i} \nonumber \\
&= (C'\otimes D') \bigotimes_{i=1}^n \ket{\Psi_i}_{B_i \bar{B}_i},
\end{align}
where $\Lambda = \bigotimes_i \lambda^i$ and $\bar{\Lambda}=\bigotimes_i \bar{\lambda}^i$. Additionally, we can move an operator $\lambda^i$ onto the $\bar{B}_i$ Hilbert space using the transpose rule. We could also choose operators $\Lambda$ and $\bar{\Lambda}$ which are not product. In this case we can no longer write $\bigotimes_i \ket{\Psi_i}$ for the state acted on by $C$ and $D$. 

The general expression for $\ket{\psi}$ without placing assumptions on the form of the projecting state is
\begin{align} \label{eq:psifromgamma}
\ket{\psi}_{A\bar{A}} &= (C\otimes D) \ket{\Psi}_{B \bar{B}},
\end{align}
where we label the Hilbert space $\bigotimes_i B_i$ by $B$ and $\bigotimes_i \bar{B}_i$ by $\bar{B}$. We give the graphical description of this expression in figure \ref{fig:cutmapprocedure}a. Eq. \ref{eq:psifromgamma} expresses the boundary state as the output of two operators acting on a state localized to the cut $\gamma$. This suggests a natural mapping to the cut, 
\begin{align} \label{eq:stateoncut}
\ket{\gamma}_{B \bar{B}}=C^\dagger C \otimes D^\dagger D \ket{\Psi}_{B \bar{B}}.
\end{align} 
This expression for the state on a cut is given graphically as figure \ref{fig:cutmapprocedure}b. 

\begin{figure}
\begin{subfigure}[b]{.5\textwidth}
  \centering
\begin{tikzpicture}[scale=0.3]

\draw[thick,postaction={on each segment={mid arrow}}]  (-3,-6)--(-3,-3.75);
\draw[thick,postaction={on each segment={mid arrow}}]  (0,-6)--(0,-3.75);
\draw[thick,postaction={on each segment={mid arrow}}] (3,-6)-- (3,-3.75);
\draw[thick,postaction={on each segment={mid arrow}}]  (-3,0)--(-3,-2.5);
\draw[thick,postaction={on each segment={mid arrow}}]  (0,0)--(0,-2.5);
\draw[thick,postaction={on each segment={mid arrow}}] (3,0)-- (3,-2.5);
\draw[thick,fill=white] (-5,-2) -- (5,-2) -- (5,-4) -- (-5,-4) -- (-5,-2);
\node at (0,-3) {$\ket{\Psi}$};

\foreach \x in {0,...,3}
	\draw[thick,postaction={on each segment={mid arrow}}] (30+40*\x:7) -- (30+40*\x:5);

\begin{scope}[shift={(0,-6)}]
\foreach \x in {0,...,3}
	\draw[thick,postaction={on each segment={mid arrow}}] (30+40*\x:-7) -- (30+40*\x:-5);
\end{scope} 

\draw[thick,fill=white] (-5,0) to [out=90,in=180] (0,5) to [out=0,in=90] (5,0) -- (-5,0);
\node at (0,2.25) {\Huge{$D$}};

\draw[thick,fill=white] (-5,-6) to [out=-90,in=180] (0,-11) to [out=0,in=-90] (5,-6) -- (-5,-6);
\node at (0,-8.25) {\Huge{$C$}};

\node at (0,-20) { };

\end{tikzpicture}
\caption{}
\end{subfigure}
\begin{subfigure}[b]{.5\textwidth}
  \centering
\begin{tikzpicture}[scale=0.3]

\draw[thick,->] (0,12) to [out=-30,in=90] (5,6);
\draw[thick] (5,6) to [out=-90,in=30] (0,0);
\draw[thick,->] (0,12) to [out=-70,in=90] (2,6);
\draw[thick] (2,6) to [out=-90,in=70] (0,0);
\draw[thick,->] (0,12) to [out=-110,in=90] (-2,6);
\draw[thick] (-2,6) to [out=-90,in=110] (0,0);
\draw[thick,->] (0,12) to [out=-150,in=90] (-5,6);
\draw[thick] (-5,6) to [out=-90,in=150] (0,0); 

\draw[thick,->] (0,-18) to [out=30,in=-90] (5,-12);
\draw[thick] (5,-12) to [out=90,in=-30] (0,-6);
\draw[thick,->] (0,-18) to [out=70,in=-90] (2,-12);
\draw[thick] (2,-12) to [out=90,in=-70] (0,-6);
\draw[thick,->] (0,-18) to [out=110,in=-90] (-2,-12);
\draw[thick] (-2,-12) to [out=90,in=-110] (0,-6);
\draw[thick,->] (0,-18) to [out=150,in=-90] (-5,-12);
\draw[thick] (-5,-12) to [out=90,in=-150] (0,-6);

\draw[thick,postaction={on each segment={mid arrow}}] (0,14) -- (0,12);
\draw[thick,postaction={on each segment={mid arrow}}] (3,14) -- (3,12);
\draw[thick,postaction={on each segment={mid arrow}}] (-3,14) -- (-3,12);

\draw[thick,postaction={on each segment={mid arrow}}] (0,-20) -- (0,-18);
\draw[thick,postaction={on each segment={mid arrow}}] (3,-20) -- (3,-18);
\draw[thick,postaction={on each segment={mid arrow}}] (-3,-20) -- (-3,-18);

\draw[thick,postaction={on each segment={mid arrow}}]  (-3,-6)--(-3,-3.75);
\draw[thick,postaction={on each segment={mid arrow}}]  (0,-6)--(0,-3.75);
\draw[thick,postaction={on each segment={mid arrow}}] (3,-6)-- (3,-3.75);
\draw[thick,postaction={on each segment={mid arrow}}]  (-3,0)--(-3,-2.5);
\draw[thick,postaction={on each segment={mid arrow}}]  (0,0)--(0,-2.5);
\draw[thick,postaction={on each segment={mid arrow}}] (3,0)-- (3,-2.5);
\draw[thick,fill=white] (-5,-2) -- (5,-2) -- (5,-4) -- (-5,-4) -- (-5,-2);
\node at (0,-3) {$\ket{\Psi}$};

\draw[thick,fill=white] (-5,0) to [out=90,in=180] (0,5) to [out=0,in=90] (5,0) -- (-5,0);
\node at (0,2.25) {\Huge{$D$}};

\draw[thick,fill=white] (-5,-6) to [out=-90,in=180] (0,-11) to [out=0,in=-90] (5,-6) -- (-5,-6);
\node at (0,-8.25) {\Huge{$C$}};

\draw[thick,fill=white] (-5,12) to [out=-90,in=180] (0,7) to [out=0,in=-90] (5,12) -- (-5,12);
\node at (0,10) {\Huge{$D^\dagger$}};

\draw[thick,fill=white] (-5,-18) to [out=90,in=180] (0,-13) to [out=0,in=90] (5,-18) -- (-5,-18);
\node at (0,-15.75) {\Huge{$C^\dagger$}};
\end{tikzpicture}
\caption{}
\end{subfigure}
\caption{(a) Graphical description of \ref{eq:psifromgamma}, which gives a tensor network state in terms of the two block tensors $C$ and $D$ defined by a cut $\gamma$. (b) Graphical description of eq \ref{eq:stateoncut}, which computes the state on a cut $\gamma$.}
\label{fig:cutmapprocedure}
\end{figure}
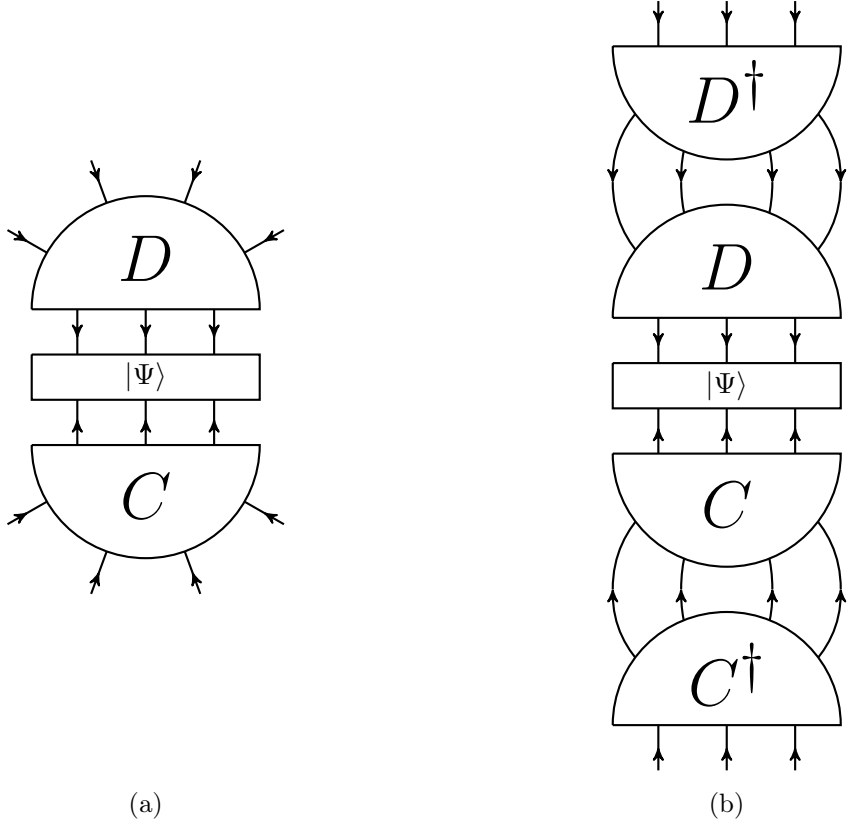 

It will be useful in generalizing away from the static case to remind ourselves of the HaPPY network construction. The key result of this earlier work is that
\begin{align} \label{eq:RT}
S(A) = \underset{\gamma^A}{\min} \, L_G(\gamma^A),
\end{align}
Where $A$ is a single boundary interval and the minimization is taken over cuts $\gamma_A$ enclosing $A$. To show this, the authors show that both sides of a minimal cut can be interpreted as a unitary circuit from the cut legs and a subset of the boundary legs to the remainder of the boundary legs. We restate this result in a slightly changed language as follows.

\begin{theorem} \label{thm:happy}
In a HaPPY network, a cut which is anchored on a boundary interval $A$ and crosses a minimal number of legs defines a map from the cut legs to the interval $A$ which is an isometry.
\end{theorem}

We will refer to cuts of a network that define isometries on both sides as \emph{isometric cuts}. This result allows us to calculate the state defined on a cut as given in eq. \ref{eq:stateoncut} whenever the cut is minimal. Recalling that all of the contractions in a HaPPY network are performed with maximally entangled pairs, \ref{eq:stateoncut} becomes
\begin{align} 
\ket{\gamma}_{B \bar{B}} = (C^\dagger C\otimes D^\dagger D) \bigotimes_{i=1}^n \ket{\Psi^+}_{B_i \bar{B}_i}.
\end{align}
When the cut is minimal theorem \ref{thm:happy} gives $C^\dagger C=\mathcal{I}$ and $D^\dagger D=\mathcal{I}$, so the state on the cut is just a collection of maximally entangled pairs, with one pair for each edge cut by $\gamma$. 

In fact, we can straightforwardly extend theorem \ref{thm:happy} to an if and only if statement as follows.

\begin{theorem}\label{thm:converse}
In a HaPPY network, a cut $\gamma$ enclosing a single boundary interval defines an isometry on both sides if and only if $\log \dim \gamma$ is minimal.
\end{theorem}

\begin{proof}
That a cut being minimal implies the maps it defines are isometries is given as theorem \ref{thm:happy}. 

Next we show that an isometric cut is minimal. Consider the boundary state as written in \ref{eq:psifromgamma}, which corresponds to the diagram in figure \ref{fig:cutmapprocedure}a. To form the density matrix on a region $A$ we draw an arrow reversed duplicate of \ref{fig:cutmapprocedure}a, and contract the $\bar{A}$ legs. Then since $D^\dagger D = \mathcal{I}$ and $\ket{\Psi} = \bigotimes_i \ket{\Psi^+}$ we are left with
\begin{align}
\rho_A = C C^\dagger.
\end{align}
To get the normalization factor note that $\tr(CC^\dagger) = \tr (C^\dagger C) = \log \dim \gamma$. Further, since $CC^\dagger$ is a projector its non-zero eigenvalues are equal to one. Using these two facts we have that
\begin{align} \label{eq:bb}
S(\rho_A) = \log \dim \gamma.
\end{align}
At the same time, the bound \ref{eq:rankbound} gives that $S(\rho_A) \leq \log\dim \gamma'$ for any cut $\gamma'$ in the network. Combining this with \ref{eq:bb} we have
\begin{align}
\log \dim \gamma \leq \log \dim \gamma'
\end{align}
for any cut in the network. Thus any cut which defines an isometry on both sides is minimal.\end{proof}

As a consequence of theorem \ref{thm:converse} the RT formula for HaPPY networks can be restated as
\begin{align} \label{eq:newRT}
S(A) = L(\gamma_{iso}^A).
\end{align}
Here $\gamma$ is any isometric cut enclosing $A$. We will refer to this as the \emph{isometric cut formula}\footnote{As with theorems \ref{thm:happy} and \ref{thm:converse} the isometric cut formula is proven only for boundary regions consisting of a single interval.}. As we will see later, this form for the Ryu-Takayangi formula is more easily generalized to dynamic states than the statement using minimization.  

\section{Length and extremal curves in tensor networks} \label{sec:graphlengthanddynamics}

\subsection{Lessons from the maximin formula} \label{sec:maximinlessons}

The maximin formula states that, for holographic states, the entropy of a boundary region $A$ can be calculated as
\begin{align} \label{eq:maximin}
S(A) = \underset{\Sigma}{\text{max}} \left(\underset{\gamma_A}{\min}\,\, L(\gamma_A)  \right).
\end{align}
That is consider a spacelike slice of the boundary and a subset $A$ of this slice. On each spacelike surface $\Sigma$ which has the chosen boundary, calculate the length of the minimal surface homologous to $A$. From the set of all those lengths choose the largest element. The resulting length will give the entropy $S(A)$. 

To translate this statement into tensor network language it is clear that we need to think of a boundary state as associated with a set of networks. Indeed, as mentioned preceding equation \ref{eq:slickpsifree}, there are various ways we can modify a network while preserving the boundary state. Beginning with a defining network these transformations give a set of networks corresponding to a single boundary state. We will explore the possibility that a set of networks generated in this way can be searched over to calculate boundary entropies, analogous to optimization over spacelike slices in the maximin formula. 

However, suppose that we have decided on a set of networks and that the maximin formula is true for these networks and their boundary state. Then the maximization step of the maximin formula gives that the minimal lengths in each network are bounded above by the entropy,
\begin{align} \label{eq:lengthentropy}
\underset{\gamma_A}{\text{min}} \, L(\gamma_A) \leq S(A).
\end{align}
This is a key inequality restricting the possible definitions of the length $L(\gamma_A)$ in the network. Indeed, suppose that we took $L(\gamma)=L_G(\gamma)$, the graph length. Then the rank bound on the entanglement entropy given in \ref{eq:rankbound} says that
\begin{align}
S(A) \leq \underset{\gamma_A}{\text{min}}\, (\log(\dim \gamma_A)) = \underset{\gamma_A}{\text{min}} \, L_G(\gamma).
\end{align}
This is the opposite inequality to \ref{eq:lengthentropy}, so we have that $S(A) = \underset{\gamma_A}{\text{min}} \, L_G(\gamma)$ for all networks in the set optimized over. This means that every network in the set must contain the extremal curve anchored on $A$. Repeating this for each of the possible boundary regions, we would conclude that every network in the set must contain the extremal curves for each boundary region. In a dynamic spacetime however no one slice should contain all of the extremal curves. Taking the graph length then prevents any description of the geometry of slices other than the constant time slices of static spacetimes.

We see that a requirement for describing spacelike slices of dynamic geometries using tensor networks is a new definition of length in networks. With a definition of length in hand, one approach is to determine the extremal cut by variation over the set of networks. However, it turns out to be simpler to generalize the isometric cut formula \ref{eq:newRT} than to try and generalize the statement \ref{eq:RT} of RT in terms of minimal cuts. Indeed, a possible generalization of \ref{eq:newRT} is just \ref{eq:newRT} again, with the modification that the isometric cut can now be chosen from within a set of networks. We find in the next section that there is a simple way to define $L(\gamma)$ that has reasonable geometric properties and which extends the isometric cut formula to the dynamic setting.

\subsection{A definition of length in tensor networks} \label{sec:lengthdefined}

From our analysis of the maximin formula we know that the definition of length will need to be changed. We claim that there is a simple way to define $L(\gamma)$ which guarantees $S(A) = L(\gamma_{iso}^A)$ whenever such an isometric cut exists in the set of networks associated with the boundary state. To see this first calculate the boundary state on $A$ in terms of the operators defined by an isometric cut $\gamma_{iso}^A$. This is most easily done by looking again at figure \ref{fig:cutmapprocedure} and considering an arrow reversed duplicate of the network in figure \ref{fig:cutmapprocedure}a. We then contract the $\bar{A}$ legs and use that $D^\dagger D = \mathcal{I}$, which yields
\begin{align}
\rho_A = C \tr_{\bar{B}} (\ketbra{\Psi}{\Psi}) C^\dagger.
\end{align}
Since $C$ is an isometry, we have that $S(\rho_A) = S(\tr_{\bar{B}}\ketbra{\Psi}{\Psi})$. Next, consider the length $L(\gamma_{iso}^A)$. As discussed in section \ref{sec:stateoncut} any cut has a state associated with it, given by \ref{eq:stateoncut}. In particular since $\gamma_{iso}^A$ is an isometric cut we have
\begin{align} \label{eq:isostate}
\ket{\gamma_{iso}}= \ket{\Psi}_{B \bar{B}}.
\end{align}
From this it is clear that defining the length as the entropy of one side of $\ket{\Psi}$ would correctly compute $S(A)$. We prefer to write this more symmetrically as the mutual information 
\begin{align} \label{eq:mutualdef}
L(\gamma_{iso}^A) = \frac{1}{2}I_{\ket{\gamma}}(B:\bar{B}).
\end{align}
For a HaPPY network, the state $\ket{\Psi}$ is a product of maximally entangled pairs and the length of a minimal cut reduces to the graph length. What about an arbitrary cut in a HaPPY network? In this case we make use of the fact that when $\ket{\Psi}$ is product, the length of a minimal cut becomes a sum over each leg in the cut
\begin{align} 
L(\gamma_{iso}^A) = \sum_i \frac{1}{2}I_{\ket{\gamma}}(B_i:\bar{B}_i).
\end{align}
Thus it is natural to associate a length to each leg individually,
\begin{align} \label{eq:indivlength}
L(\gamma_i)=\frac{1}{2}I_{\ket{\gamma}}(B_i:\bar{B}_i).
\end{align}
For an arbitrary curve we can define its length to be the sum of the lengths of each leg, where we calculate the length of a single leg by finding an isometric cut containing that leg. In the HaPPY network this assigns all legs a length of $\log \dim \gamma_i$.

For a non-HaPPY network we can attempt to assign lengths to every leg by the same procedure of looking at the product factors of the isometric cuts which are in that network. In general however not every leg will be part of an isometric cut, and it may not be possible to assign a length to every leg. This manner of building up the lengths of arbitrary cuts using the lengths of isometric cuts is reminiscent of the differential entropy formula \cite{balasubramanian2014bulk, headrick2014holographic}. Extremal cuts in the differential entropy formula play the role of isometric cuts in the procedure outlined here. This is consistent with our interpretation of \ref{eq:newRT} as applying to dynamic spacetimes, since the isometric cuts of \ref{eq:newRT} are playing the role of extremal curves.

As a basic check on this definition of length, we should confirm that all isometric cuts passing a single leg will assign the same value of length to that leg. Indeed, for any isometric cut crossing a segment $\gamma_i$ the length of that segment is given by the mutual information $I(B:\bar{B})/2$ computed in the projecting state $\ket{\Psi_i}$, and is independent of the operators $C$ and $D$ defined by whichever isometric cut has been chosen. Further, the projecting state $\ket{\Psi_i}$ is fixed for a given network. 

We can also notice that contracting the network in a different basis has no effect on the length of a cut. This follows from the invariance of the mutual information under local unitaries. Finally, it would be nice to see that if a cut $\gamma$ is composed of two segments $\gamma_1$ and $\gamma_2$ then 
\begin{align}\label{eq:addit}
L(\gamma) = L(\gamma_1) + L(\gamma_2).
\end{align}
We'll refer to this as the additivity property. The additivity property is not guaranteed by our definition of length, but rather depends on the structure of the state $\ket{\Psi}$ appearing in \ref{eq:isostate}. For example if this state is product across each leg, that is if
\begin{align}
\ket{\gamma_{iso}} = \bigotimes_{i=1}^n \ket{\Psi_i},
\end{align}
then the length is additive at the level of individual legs. However, in other cases it may happen that entanglement is present across the $B_i$, in which case the length will not be additive across individual legs. In the dynamic example given in section \ref{sec:dynamicexample} we will allow legs which are contracted with a common vertex to share entanglement, meaning additivity may fail at the level of a small number (in the case there, three) legs.


\subsection{A static example} \label{sec:staticexample}

As an illustration of this assignment of length we look at a network which does not satisfy RT when the graph length is used, but does when using the mutual information based definition. Our example is based on the network shown in figure \ref{fig:sixlegexample}a. The six legged tensors are perfect tensors, and the two legged tensors shown as solid black dots are maximally entangled pairs used to form the contraction. This is a HaPPY network and the boundary entropies are all given by the minimal number of legs cut to separate off a boundary region. 

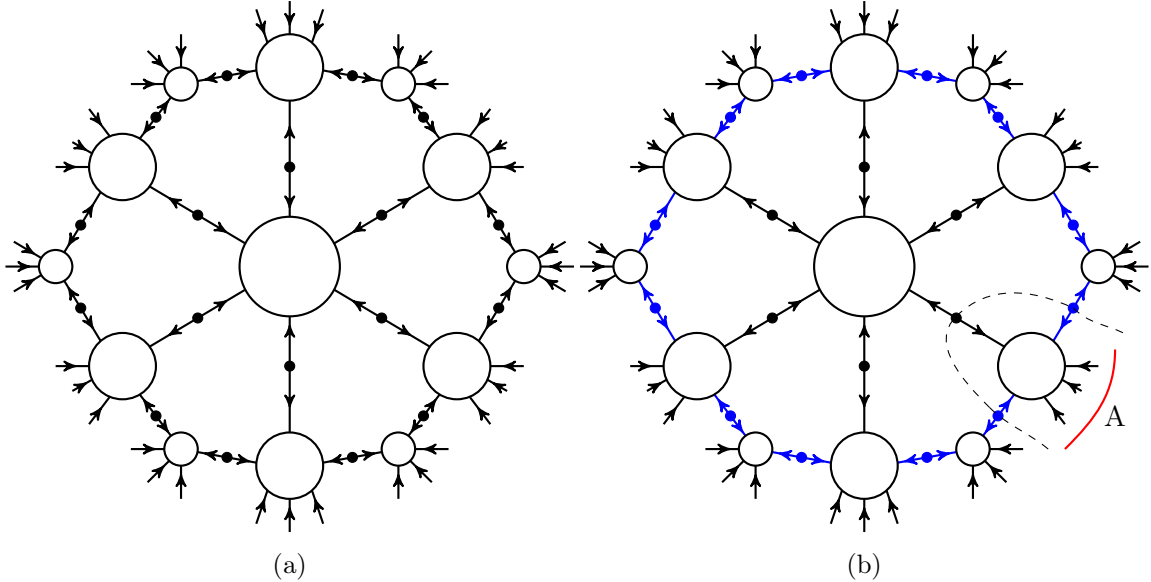
\begin{figure}
\centering
\begin{subfigure}{.5\textwidth}
  \centering
    \begin{tikzpicture}[scale=0.22]
\bellbrasegment{0}{1}{0}{11}{0.3}
\bellbrasegment{0}{-1}{0}{-11}{0.3}
\bellbrasegment{2}{1}{9}{5.2}{0.3}
\bellbrasegment{2}{-1}{9}{-5.2}{0.3}
\bellbrasegment{-2}{1}{-9}{5.2}{0.3}
\bellbrasegment{-2}{-1}{-9}{-5.2}{0.3}
\bellbrasegment{11}{5}{14}{0}{0.3}
\bellbrasegment{11}{-5}{14}{0}{0.3}

\bellbrasegment{-11}{5}{-14}{0}{0.3}
\bellbrasegment{-11}{-5}{-14}{0}{0.3}

\bellbrasegment{-9.5}{7}{-6.5}{11}{0.3}
\bellbrasegment{9.5}{7}{6.5}{11}{0.3}
\bellbrasegment{-9.5}{-7}{-6.5}{-11}{0.3}
\bellbrasegment{9.5}{-7}{6.5}{-11}{0.3}

\bellbrasegment{-1}{12}{-6.5}{11}{0.3}
\bellbrasegment{1}{12}{6.5}{11}{0.3}
\bellbrasegment{-1}{-12}{-6.5}{-11}{0.3}
\bellbrasegment{1}{-12}{6.5}{-11}{0.3}

\arrowedsegment{14}{6}{11.5}{6}
\polararrowedsegment{13}{7.5}{210}{2}
\arrowedsegment{12}{9.5}{10.5}{7.5}

\arrowedsegment{-14}{6}{-11.5}{6}
\polararrowedsegment{-13}{7.5}{-30}{2}
\arrowedsegment{-12}{9.5}{-10.5}{7.5}

\arrowedsegment{-14}{-6}{-11.5}{-6}
\polararrowedsegment{-13}{-7.5}{30}{2}
\arrowedsegment{-12}{-9.5}{-10.5}{-7.5}

\arrowedsegment{14}{-6}{11.5}{-6}
\polararrowedsegment{13}{-7.5}{150}{2}
\arrowedsegment{12}{-9.5}{10.5}{-7.5}

\arrowedsegment{0}{-16}{0}{-13}
\arrowedsegment{2}{-15.5}{1}{-12.5}
\arrowedsegment{-2}{-15.5}{-1}{-12.5}

\arrowedsegment{0}{16}{0}{13}
\arrowedsegment{2}{15.5}{1}{12.5}
\arrowedsegment{-2}{15.5}{-1}{12.5}

\arrowedsegment{17}{0}{14}{0}
\arrowedsegment{16.5}{1.5}{14}{0}
\arrowedsegment{16.5}{-1.5}{14}{0}

\arrowedsegment{-17}{0}{-14}{0}
\arrowedsegment{-16.5}{1.5}{-14}{0}
\arrowedsegment{-16.5}{-1.5}{-14}{0}

\arrowedsegment{9.5}{11}{6.5}{11}
\arrowedsegment{6.5}{14}{6.5}{11}
\arrowedsegment{8.5}{13}{6.5}{11}

\arrowedsegment{-9.5}{11}{-6.5}{11}
\arrowedsegment{-6.5}{14}{-6.5}{11}
\arrowedsegment{-8.5}{13}{-6.5}{11}

\arrowedsegment{9.5}{-11}{6.5}{-11}
\arrowedsegment{6.5}{-14}{6.5}{-11}
\arrowedsegment{8.5}{-13}{6.5}{-11}

\arrowedsegment{-9.5}{-11}{-6.5}{-11}
\arrowedsegment{-6.5}{-14}{-6.5}{-11}
\arrowedsegment{-8.5}{-13}{-6.5}{-11}

\draw[thick,fill=white] (0,0) circle [radius=3];

\draw[thick,fill=white] (10,6) circle [radius=2];
\draw[thick,fill=white] (10,-6) circle [radius=2];
\draw[thick,fill=white] (-10,6) circle [radius=2];
\draw[thick,fill=white] (-10,-6) circle [radius=2];
\draw[thick,fill=white] (0,12) circle [radius=2];
\draw[thick,fill=white] (0,-12) circle [radius=2];

\draw[thick,fill=white] (14,0) circle [radius=1];
\draw[thick,fill=white] (-14,0) circle [radius=1];
\draw[thick,fill=white] (6.5,11) circle [radius=1];
\draw[thick,fill=white] (-6.5,11) circle [radius=1];
\draw[thick,fill=white] (6.5,-11) circle [radius=1];
\draw[thick,fill=white] (-6.5,-11) circle [radius=1];

\end{tikzpicture}
\caption{}
\end{subfigure}%
\hfill
\begin{subfigure}{.5\textwidth}
  \centering
  \begin{tikzpicture}[scale=0.22]

\bellbrasegment{0}{1}{0}{11}{0.3}
\bellbrasegment{0}{-1}{0}{-11}{0.3}
\bellbrasegment{2}{1}{9}{5.2}{0.3}
\bellbrasegment{2}{-1}{9}{-5.2}{0.3}
\bellbrasegment{-2}{1}{-9}{5.2}{0.3}
\bellbrasegment{-2}{-1}{-9}{-5.2}{0.3}

\bluebrasegment{11}{5}{14}{0}{0.3}
\bluebrasegment{11}{-5}{14}{0}{0.3}

\bluebrasegment{-11}{5}{-14}{0}{0.3}
\bluebrasegment{-11}{-5}{-14}{0}{0.3}

\bluebrasegment{-9.5}{7}{-6.5}{11}{0.3}
\bluebrasegment{9.5}{7}{6.5}{11}{0.3}
\bluebrasegment{-9.5}{-7}{-6.5}{-11}{0.3}
\bluebrasegment{9.5}{-7}{6.5}{-11}{0.3}

\bluebrasegment{-1}{12}{-6.5}{11}{0.3}
\bluebrasegment{1}{12}{6.5}{11}{0.3}
\bluebrasegment{-1}{-12}{-6.5}{-11}{0.3}
\bluebrasegment{1}{-12}{6.5}{-11}{0.3}

\arrowedsegment{14}{6}{11.5}{6}
\polararrowedsegment{13}{7.5}{210}{2}
\arrowedsegment{12}{9.5}{10.5}{7.5}

\arrowedsegment{-14}{6}{-11.5}{6}
\polararrowedsegment{-13}{7.5}{-30}{2}
\arrowedsegment{-12}{9.5}{-10.5}{7.5}

\arrowedsegment{-14}{-6}{-11.5}{-6}
\polararrowedsegment{-13}{-7.5}{30}{2}
\arrowedsegment{-12}{-9.5}{-10.5}{-7.5}

\arrowedsegment{14}{-6}{11.5}{-6}
\polararrowedsegment{13}{-7.5}{150}{2}
\arrowedsegment{12}{-9.5}{10.5}{-7.5}

\arrowedsegment{0}{-16}{0}{-13}
\arrowedsegment{2}{-15.5}{1}{-12.5}
\arrowedsegment{-2}{-15.5}{-1}{-12.5}

\arrowedsegment{0}{16}{0}{13}
\arrowedsegment{2}{15.5}{1}{12.5}
\arrowedsegment{-2}{15.5}{-1}{12.5}

\arrowedsegment{17}{0}{14}{0}
\arrowedsegment{16.5}{1.5}{14}{0}
\arrowedsegment{16.5}{-1.5}{14}{0}

\arrowedsegment{-17}{0}{-14}{0}
\arrowedsegment{-16.5}{1.5}{-14}{0}
\arrowedsegment{-16.5}{-1.5}{-14}{0}

\arrowedsegment{9.5}{11}{6.5}{11}
\arrowedsegment{6.5}{14}{6.5}{11}
\arrowedsegment{8.5}{13}{6.5}{11}

\arrowedsegment{-9.5}{11}{-6.5}{11}
\arrowedsegment{-6.5}{14}{-6.5}{11}
\arrowedsegment{-8.5}{13}{-6.5}{11}

\arrowedsegment{9.5}{-11}{6.5}{-11}
\arrowedsegment{6.5}{-14}{6.5}{-11}
\arrowedsegment{8.5}{-13}{6.5}{-11}

\arrowedsegment{-9.5}{-11}{-6.5}{-11}
\arrowedsegment{-6.5}{-14}{-6.5}{-11}
\arrowedsegment{-8.5}{-13}{-6.5}{-11}

\draw[red,thick] (12,-11) to [out=45,in=-90] (15,-5);
\node at (15,-9) {A};

\draw[dashed] (11,-11) to [out=140,in=-45] (7.5,-8.5);
\draw[dashed] (7.5,-8.5) to [out=135,in=-135] (5.5,-3);
\draw[dashed] (5.5,-3) to [out=45,in=155] (12.5,-2.5);
\draw[dashed] (12.5,-2.5) to [out=-45,in=155] (15.5,-4);

\draw[thick,fill=white] (0,0) circle [radius=3];

\draw[thick,fill=white] (10,6) circle [radius=2];
\draw[thick,fill=white] (10,-6) circle [radius=2];
\draw[thick,fill=white] (-10,6) circle [radius=2];
\draw[thick,fill=white] (-10,-6) circle [radius=2];
\draw[thick,fill=white] (0,12) circle [radius=2];
\draw[thick,fill=white] (0,-12) circle [radius=2];

\draw[thick,fill=white] (14,0) circle [radius=1];
\draw[thick,fill=white] (-14,0) circle [radius=1];
\draw[thick,fill=white] (6.5,11) circle [radius=1];
\draw[thick,fill=white] (-6.5,11) circle [radius=1];
\draw[thick,fill=white] (6.5,-11) circle [radius=1];
\draw[thick,fill=white] (-6.5,-11) circle [radius=1];

\end{tikzpicture}
\caption{}
\end{subfigure}
\caption{(a) A HaPPY network. Six legged vertices are perfect tensors and two legged projecting pairs are maximally entangled states. All boundary entropies are given by the graph length of a minimal cut in the network. (b) A network which satisfies the Ryu-Takayanagi formula using the mutual information based definition of length, but does not satisfy Ryu-Takayanagi when using the graph length. The blue dots and legs represent the projecting state given in eq. \ref{eq:newprojector}. The dashed line represents the isometric cut for the boundary region $A$.}
\label{fig:sixlegexample}
\end{figure}

For our example we replace the maximally entangled pairs around the edge of the network with another state, $\ket{\Psi_i}$, which for convenience we write in the form\footnote{It is always possible to write $\ket{\Psi}$ in this way because we can write $\ket{\Psi_i} = A \otimes B \ket{\Psi^+}$
and then use the transpose rule to move $B$ to the other subspace.}
\begin{align} \label{eq:newprojector}
\ket{\Psi_i}=(\mathcal{O} \otimes \mathcal{I}) \ket{\Psi^+}.
\end{align}
The modified network is shown in figure \ref{fig:sixlegexample}b. We claim that this network satisfies the RT formula using our new definition of length in terms of mutual information, but not using the graph length. To see this we begin by computing the entropy of the three boundary legs marked region A.

\begin{figure}

\begin{tikzpicture}[scale=0.2];
\draw[thick,postaction={on each segment={mid arrow}}] (0,0) -- (0,18);
\draw[thick,postaction={on each segment={mid arrow}}] (0,0) to [out=60,in=-60] (0,18);
\draw[thick,postaction={on each segment={mid arrow}}] (0,0) to [out=120,in=-120] (0,18);
\draw[thick,postaction={on each segment={mid arrow}}] (23,0) to [out=45,in=-45] (23,18);
\draw[thick,postaction={on each segment={mid arrow}}] (23,0) to [out=75,in=-75] (23,18);
\draw[thick,postaction={on each segment={mid arrow}}] (23,0) to [out=105,in=-105] (23,18);
\draw[thick,postaction={on each segment={mid arrow}}] (23,0) to [out=135,in=-135] (23,18);
\bluebrasegment{0}{18}{23}{18}{0.75}
\blueketsegment{0}{0}{23}{0}{0.75}
\arrowedsegment{0}{-3}{0}{-6}
\arrowedsegment{23}{-2}{23}{-6}
\arrowedsegment{0}{24}{0}{21}
\arrowedsegment{23}{24}{23}{20}
\arrowedsegment{-3}{0}{-6}{0}
\arrowedsegment{-6}{18}{-3}{18}
\arrowedsegment{20}{0}{24}{0}
\arrowedsegment{24}{18}{20}{18}
\arrowedsegment{29}{18}{25}{18}
\arrowedsegment{25}{0}{29}{0}
\draw[fill=white] (0,0) circle [radius=3];
\draw[fill=white] (0,18) circle [radius =3];
\draw[fill=white] (23,0) circle [radius=2];
\draw[fill=white] (23,18) circle [radius=2];
\node at (33,9) {$=$};

\draw[thick,postaction={on each segment={mid arrow}}] (40,18) to [out=0,in=0] (40,0);
\draw[thick,postaction={on each segment={mid arrow}}] (48,0) -- (48,18);

\draw[fill=blue] (56,0) circle [radius=0.75];
\draw[fill=blue] (56,18) circle [radius=0.75];

\draw[blue, thick, postaction={on each segment={mid arrow}}] (56,18) to [out=0,in=0] (56,0);
\draw[blue, thick, postaction={on each segment={mid arrow}}] (56,0) to [out=180,in=180] (56,18);
\draw[thick, postaction={on each segment={mid arrow}}] (69,0) to [out=180,in=180] (69,18);

\draw[dashed,thick] (69,18) -- (72,18);
\draw[dashed,thick] (69,0) -- (72,0);

\draw[dashed,thick] (40,18) -- (37,18);
\draw[dashed,thick] (40,0) -- (37,0);

\end{tikzpicture}
\caption{The basic simplification used to compute the density matrix of region A in figure \ref{fig:sixlegexample}b. The six leg tensors are from the edge of the network shown in figure \ref{fig:sixlegexample}b. The effect of the state given in equation \ref{eq:newprojector} is to add a normalization factor, represented as the blue loop at right.}
\label{fig:densityAcontraction}
\end{figure}
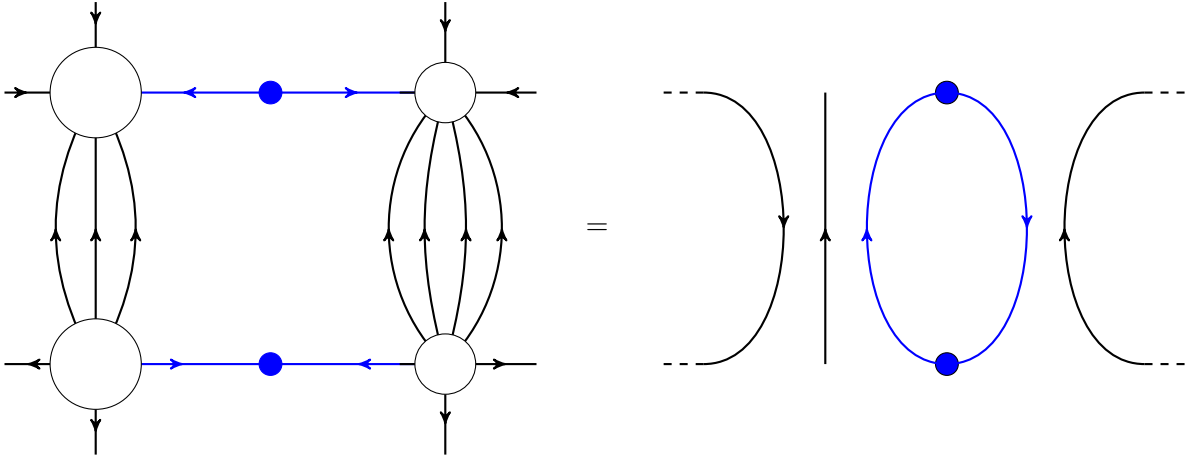

We can go a long ways towards computing this entropy using the graphical notation. To do this we draw an arrow reversed copy of figure \ref{fig:sixlegexample}b, then contract all the legs in $\bar{A}$. To understand what happens when this is done consider the diagram in figure \ref{fig:densityAcontraction}. The simplification shown there gives that all the insertions of $\mathcal{O}$ that are not adjacent to region $A$ turn into pure normalization factors. After continuing the contraction we are left with the density matrix illustrated in figure \ref{fig:densityAsixleg}. As an operator expression, this is
\begin{align} \label{eq:rhoa}
\rho_A = T (\mathcal{O}\mathcal{O}^\dagger\otimes \mathcal{I}\otimes\mathcal{O}\mathcal{O}^\dagger) T^\dagger.
\end{align}
The entropy is given by:
\begin{align} 
S(\rho_A) = S(T^\dagger \rho_A T) = 2\cdot S(\mathcal{O}\mathcal{O}^\dagger) + 1
\end{align}
By choosing $\mathcal{O}$ to be non-unitary we find an entropy less than $3=L_G(\gamma_{min})$, so we have that the RT formula using the graph length fails. 

\begin{figure}
\begin{center}
\begin{tikzpicture}[scale=0.3]
\draw[thick,blue,postaction={on each segment={mid arrow}}] (2,2) to [out=45,in=180] (7,5);
\draw[fill=blue] (7,5) circle [radius=0.75];
\draw[thick,blue,postaction={on each segment={mid arrow}}] (13,5) to [out=180,in=0] (7,5);
\draw[fill=blue] (13,5) circle [radius=0.75];
\draw[thick,blue,postaction={on each segment={mid arrow}}] (13,5) to [out=0,in=135] (18,2);
\draw[thick,blue,postaction={on each segment={mid arrow}}] (2,-2) to [out=-45,in=180] (7,-5);
\draw[fill=blue] (7,-5) circle [radius=0.75];
\draw[thick,blue,postaction={on each segment={mid arrow}}] (13,-5) to [out=180,in=0] (7,-5);
\draw[fill=blue] (13,-5) circle [radius=0.75];
\draw[thick,blue,postaction={on each segment={mid arrow}}] (13,-5) to [out=0,in=225] (18,-2);
\draw[thick,postaction={on each segment={mid arrow}}] (0,0) -- (20,0);
\draw[thick,postaction={on each segment={mid arrow}}] (-3,0) -- (-7,0);
\draw[thick,postaction={on each segment={mid arrow}}] (-2,2) -- (-4,5);
\draw[thick,postaction={on each segment={mid arrow}}] (-2,-2) -- (-4,-5);
\draw[thick,postaction={on each segment={mid arrow}}] (27,0) -- (23,0);
\draw[thick,postaction={on each segment={mid arrow}}] (24,5)--(22,2);
\draw[thick,postaction={on each segment={mid arrow}}] (24,-5)--(22,-2);
\draw[fill=white] (0,0) circle [radius=3];
\draw[fill=white] (20,0) circle [radius=3];
\end{tikzpicture}
\end{center}
\caption{Graphical representation of the density matrix of the region $A$ from figure \ref{fig:sixlegexample}b. }
\label{fig:densityAsixleg}
\end{figure}
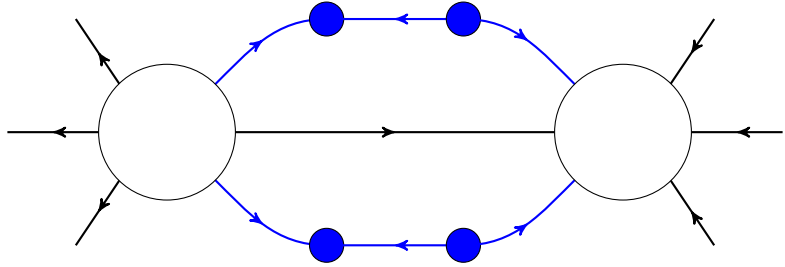

What is the minimal length when computed using eq. \ref{eq:mutualdef}? The cut with the minimal number of legs actually defines an isometry on both sides\footnote{To be precise, the $\bar{A}$ side of this cut has the property $D_\gamma^\dagger D_\gamma = \alpha \mathcal{I}$ for a scalar $\alpha$. This scalar is divided out when the normalization is added to the state.}. This is in fact what we used when showing that the reduced density matrix $\rho_A$ was given by expression \ref{eq:rhoa}. Since both sides of the cut are isometries the state on the cut is just given by the projecting state, which in this case is
\begin{align}
\ket{\Psi} = \ket{\Psi_i}_{B_1 \bar{B}_1} \otimes \ket{\Psi^+}_{B_2 \bar{B}_2} \otimes \ket{\Psi_i}_{B_3 \bar{B}_3}.
\end{align}
There are two legs with operator insertions, which have length
\begin{align}
L = \frac{1}{2}I_{\Psi_i}(B:\bar{B})= S(\mathcal{O}\mathcal{O}^\dagger),
\end{align}
while the leg with no insertion has length $I_{\Psi^+}(B:\bar{B})/2=1$, giving $L = 2\cdot S(\mathcal{O}\mathcal{O}^\dagger)+1 = S(\rho_A)$. It is straightforward to check the minimal lengths and boundary entropies of any other region $A$ in the network shown in figure \ref{fig:sixlegexample}b agree.

\begin{figure}
\begin{center}
\begin{subfigure}[b]{0.45\textwidth}
\begin{center}
\begin{tikzpicture}[scale=0.187]
\bluebrasegment{0}{1}{0}{11}{0.3}
\bellbrasegment{0}{-1}{0}{-11}{0.3}
\bluebrasegment{2}{1}{9}{5.2}{0.3}
\bellbrasegment{2}{-1}{9}{-5.2}{0.3}
\bluebrasegment{-2}{1}{-9}{5.2}{0.3}
\bellbrasegment{-2}{-1}{-9}{-5.2}{0.3}
\bellbrasegment{11}{5}{14}{0}{0.3}
\bellbrasegment{11}{-5}{14}{0}{0.3}
\bellbrasegment{-11}{5}{-14}{0}{0.3}
\bellbrasegment{-11}{-5}{-14}{0}{0.3}
\bellbrasegment{-9.5}{7}{-6.5}{11}{0.3}
\bellbrasegment{9.5}{7}{6.5}{11}{0.3}
\bellbrasegment{-9.5}{-7}{-6.5}{-11}{0.3}
\bellbrasegment{9.5}{-7}{6.5}{-11}{0.3}
\bellbrasegment{-1}{12}{-6.5}{11}{0.3}
\bellbrasegment{1}{12}{6.5}{11}{0.3}
\bellbrasegment{-1}{-12}{-6.5}{-11}{0.3}
\bellbrasegment{1}{-12}{6.5}{-11}{0.3}
\arrowedsegment{14}{6}{11.5}{6}
\polararrowedsegment{13}{7.5}{210}{2}
\arrowedsegment{12}{9.5}{10.5}{7.5}
\arrowedsegment{-14}{6}{-11.5}{6}
\polararrowedsegment{-13}{7.5}{-30}{2}
\arrowedsegment{-12}{9.5}{-10.5}{7.5}
\arrowedsegment{-14}{-6}{-11.5}{-6}
\polararrowedsegment{-13}{-7.5}{30}{2}
\arrowedsegment{-12}{-9.5}{-10.5}{-7.5}
\arrowedsegment{14}{-6}{11.5}{-6}
\polararrowedsegment{13}{-7.5}{150}{2}
\arrowedsegment{12}{-9.5}{10.5}{-7.5}
\arrowedsegment{0}{-16}{0}{-13}
\arrowedsegment{2}{-15.5}{1}{-12.5}
\arrowedsegment{-2}{-15.5}{-1}{-12.5}
\arrowedsegment{0}{16}{0}{13}
\arrowedsegment{2}{15.5}{1}{12.5}
\arrowedsegment{-2}{15.5}{-1}{12.5}
\arrowedsegment{17}{0}{14}{0}
\arrowedsegment{16.5}{1.5}{14}{0}
\arrowedsegment{16.5}{-1.5}{14}{0}
\arrowedsegment{-17}{0}{-14}{0}
\arrowedsegment{-16.5}{1.5}{-14}{0}
\arrowedsegment{-16.5}{-1.5}{-14}{0}
\arrowedsegment{9.5}{11}{6.5}{11}
\arrowedsegment{6.5}{14}{6.5}{11}
\arrowedsegment{8.5}{13}{6.5}{11}
\arrowedsegment{-9.5}{11}{-6.5}{11}
\arrowedsegment{-6.5}{14}{-6.5}{11}
\arrowedsegment{-8.5}{13}{-6.5}{11}
\arrowedsegment{9.5}{-11}{6.5}{-11}
\arrowedsegment{6.5}{-14}{6.5}{-11}
\arrowedsegment{8.5}{-13}{6.5}{-11}
\arrowedsegment{-9.5}{-11}{-6.5}{-11}
\arrowedsegment{-6.5}{-14}{-6.5}{-11}
\arrowedsegment{-8.5}{-13}{-6.5}{-11}
\draw[thick,fill=white] (0,0) circle [radius=3];
\node at (0,0) {a};
\draw[thick,fill=white] (10,6) circle [radius=2];
\node at (10,6) {b};
\draw[thick,fill=white] (10,-6) circle [radius=2];
\node at (10,-6) {c};
\draw[thick,fill=white] (-10,6) circle [radius=2];
\node at (-10,6) {f};
\draw[thick,fill=white] (-10,-6) circle [radius=2];
\node at (-10,-6) {e};
\draw[thick,fill=white] (0,12) circle [radius=2];
\node at (0,12) {g};
\draw[thick,fill=white] (0,-12) circle [radius=2];
\node at (0,-12) {d};
\draw[thick,fill=white] (14,0) circle [radius=1];
\draw[thick,fill=white] (-14,0) circle [radius=1];
\draw[thick,fill=white] (6.5,11) circle [radius=1];
\draw[thick,fill=white] (-6.5,11) circle [radius=1];
\draw[thick,fill=white] (6.5,-11) circle [radius=1];
\draw[thick,fill=white] (-6.5,-11) circle [radius=1];
\draw[thick, dashed] (0,6) to [out=30,in=-100] (3.75,11.5);
\draw[thick, dashed] (3.75,11.5) to [out=90,in=-90] (3.75,16.5);
\draw[thick, dashed] (0,6) to [out=200,in=30] (-5.5,3.1);
\draw[thick, dashed] (-5.5,3.1) to [out=210,in=-30] (-12.5,2.5);
\draw[thick, dashed] (-12.5,2.5) to [out=150,in=-30] (-16.5,5);
\node[above] at (3.75,16.5) {$\gamma_1$};
\draw[thick, dashed] (0,-6) to [out=30,in=-120] (5.5,-3.1);
\draw[thick, dashed] (5.5,-3.1) to [out=60,in=-90] (5.5,3.1);
\draw[thick, dashed] (5.5,3.1) to [out=90,in=-150] (8,9);
\draw[thick, dashed] (8,9) to [out=30,in=-140] (12,12);
\node[right] at (12,12) {$\gamma_2$};
\draw[thick, dashed] (0,-6) to [out=210,in=80] (-3.75,-11.5);
\draw[thick, dashed] (-3.75,-11.5) to [out=-100,in=-100] (-4.75,-15.5);
\end{tikzpicture}
\end{center}
\caption{}
\end{subfigure}
\hfill
\begin{subfigure}[b]{.45\textwidth}
\begin{center}
\begin{tikzpicture}[scale=0.187]
\bellbrasegment{0}{1}{0}{11}{0.3}
\bluebrasegment{0}{-1}{0}{-11}{0.3}
\bluebrasegment{2}{1}{9}{5.2}{0.3}
\bluebrasegment{2}{-1}{9}{-5.2}{0.3}
\bellbrasegment{-2}{1}{-9}{5.2}{0.3}
\bellbrasegment{-2}{-1}{-9}{-5.2}{0.3}
\draw[ultra thick,blue] (0,-6) to [out=0,in=-120] (5.5,-3.1);
\draw[ultra thick,blue] (5.5,-3.1) to [out=60,in=-60] (5.5,3.1);

\draw[thick, dashed] (0,-6) to [out=30,in=-120] (5.5,-2.1);
\draw[thick, dashed] (5.5,-2.1) to [out=60,in=-90] (5.5,3.1);
\draw[thick, dashed] (5.5,3.1) to [out=90,in=-150] (8,9);
\draw[thick, dashed] (8,9) to [out=30,in=-140] (12,12);
\node[right] at (12,12) {$\gamma_2$};
\draw[thick, dashed] (0,-6) to [out=210,in=80] (-3.75,-11.5);
\draw[thick, dashed] (-3.75,-11.5) to [out=-100,in=-100] (-4.75,-15.5);





\bellbrasegment{11}{5}{14}{0}{0.3}
\bellbrasegment{11}{-5}{14}{0}{0.3}
\bellbrasegment{-11}{5}{-14}{0}{0.3}
\bellbrasegment{-11}{-5}{-14}{0}{0.3}
\bellbrasegment{-9.5}{7}{-6.5}{11}{0.3}
\bellbrasegment{9.5}{7}{6.5}{11}{0.3}
\bellbrasegment{-9.5}{-7}{-6.5}{-11}{0.3}
\bellbrasegment{9.5}{-7}{6.5}{-11}{0.3}
\bellbrasegment{-1}{12}{-6.5}{11}{0.3}
\bellbrasegment{1}{12}{6.5}{11}{0.3}
\bellbrasegment{-1}{-12}{-6.5}{-11}{0.3}
\bellbrasegment{1}{-12}{6.5}{-11}{0.3}
\arrowedsegment{14}{6}{11.5}{6}
\polararrowedsegment{13}{7.5}{210}{2}
\arrowedsegment{12}{9.5}{10.5}{7.5}
\arrowedsegment{-14}{6}{-11.5}{6}
\polararrowedsegment{-13}{7.5}{-30}{2}
\arrowedsegment{-12}{9.5}{-10.5}{7.5}
\arrowedsegment{-14}{-6}{-11.5}{-6}
\polararrowedsegment{-13}{-7.5}{30}{2}
\arrowedsegment{-12}{-9.5}{-10.5}{-7.5}
\arrowedsegment{14}{-6}{11.5}{-6}
\polararrowedsegment{13}{-7.5}{150}{2}
\arrowedsegment{12}{-9.5}{10.5}{-7.5}
\arrowedsegment{0}{-16}{0}{-13}
\arrowedsegment{2}{-15.5}{1}{-12.5}
\arrowedsegment{-2}{-15.5}{-1}{-12.5}
\arrowedsegment{0}{16}{0}{13}
\arrowedsegment{2}{15.5}{1}{12.5}
\arrowedsegment{-2}{15.5}{-1}{12.5}
\arrowedsegment{17}{0}{14}{0}
\arrowedsegment{16.5}{1.5}{14}{0}
\arrowedsegment{16.5}{-1.5}{14}{0}
\arrowedsegment{-17}{0}{-14}{0}
\arrowedsegment{-16.5}{1.5}{-14}{0}
\arrowedsegment{-16.5}{-1.5}{-14}{0}
\arrowedsegment{9.5}{11}{6.5}{11}
\arrowedsegment{6.5}{14}{6.5}{11}
\arrowedsegment{8.5}{13}{6.5}{11}
\arrowedsegment{-9.5}{11}{-6.5}{11}
\arrowedsegment{-6.5}{14}{-6.5}{11}
\arrowedsegment{-8.5}{13}{-6.5}{11}
\arrowedsegment{9.5}{-11}{6.5}{-11}
\arrowedsegment{6.5}{-14}{6.5}{-11}
\arrowedsegment{8.5}{-13}{6.5}{-11}
\arrowedsegment{-9.5}{-11}{-6.5}{-11}
\arrowedsegment{-6.5}{-14}{-6.5}{-11}
\arrowedsegment{-8.5}{-13}{-6.5}{-11}
\draw[thick,fill=white] (0,0) circle [radius=3];
\node at (0,0) {a};
\draw[thick,fill=white] (10,6) circle [radius=2];
\node at (10,6) {b};
\draw[thick,fill=white] (10,-6) circle [radius=2];
\node at (10,-6) {c};
\draw[thick,fill=white] (-10,6) circle [radius=2];
\node at (-10,6) {f};
\draw[thick,fill=white] (-10,-6) circle [radius=2];
\node at (-10,-6) {e};
\draw[thick,fill=white] (0,12) circle [radius=2];
\node at (0,12) {g};
\draw[thick,fill=white] (0,-12) circle [radius=2];
\node at (0,-12) {d};
\draw[thick,fill=white] (14,0) circle [radius=1];
\draw[thick,fill=white] (-14,0) circle [radius=1];
\draw[thick,fill=white] (6.5,11) circle [radius=1];
\draw[thick,fill=white] (-6.5,11) circle [radius=1];
\draw[thick,fill=white] (6.5,-11) circle [radius=1];
\draw[thick,fill=white] (-6.5,-11) circle [radius=1];
\end{tikzpicture}
\end{center}
\caption{}
\end{subfigure}
\begin{subfigure}[b]{.45\textwidth}
  \begin{center}

  \begin{tikzpicture}[scale=0.187]
    \reflectbox{
\bellbrasegment{0}{1}{0}{11}{0.3}
\bluebrasegment{0}{-1}{0}{-11}{0.3}
\bluebrasegment{2}{1}{9}{5.2}{0.3}
\bluebrasegment{2}{-1}{9}{-5.2}{0.3}
\bellbrasegment{-2}{1}{-9}{5.2}{0.3}
\bellbrasegment{-2}{-1}{-9}{-5.2}{0.3}
\draw[ultra thick,blue] (0,-6) to [out=0,in=-120] (5.5,-3.1);
\draw[ultra thick,blue] (5.5,-3.1) to [out=60,in=-60] (5.5,3.1);
\bellbrasegment{11}{5}{14}{0}{0.3}
\bellbrasegment{11}{-5}{14}{0}{0.3}
\bellbrasegment{-11}{5}{-14}{0}{0.3}
\bellbrasegment{-11}{-5}{-14}{0}{0.3}
\bellbrasegment{-9.5}{7}{-6.5}{11}{0.3}
\bellbrasegment{9.5}{7}{6.5}{11}{0.3}
\bellbrasegment{-9.5}{-7}{-6.5}{-11}{0.3}
\bellbrasegment{9.5}{-7}{6.5}{-11}{0.3}
\bellbrasegment{-1}{12}{-6.5}{11}{0.3}
\bellbrasegment{1}{12}{6.5}{11}{0.3}
\bellbrasegment{-1}{-12}{-6.5}{-11}{0.3}
\bellbrasegment{1}{-12}{6.5}{-11}{0.3}
\arrowedsegment{14}{6}{11.5}{6}
\polararrowedsegment{13}{7.5}{210}{2}
\arrowedsegment{12}{9.5}{10.5}{7.5}
\arrowedsegment{-14}{6}{-11.5}{6}
\polararrowedsegment{-13}{7.5}{-30}{2}
\arrowedsegment{-12}{9.5}{-10.5}{7.5}
\arrowedsegment{-14}{-6}{-11.5}{-6}
\polararrowedsegment{-13}{-7.5}{30}{2}
\arrowedsegment{-12}{-9.5}{-10.5}{-7.5}
\arrowedsegment{14}{-6}{11.5}{-6}
\polararrowedsegment{13}{-7.5}{150}{2}
\arrowedsegment{12}{-9.5}{10.5}{-7.5}
\arrowedsegment{0}{-16}{0}{-13}
\arrowedsegment{2}{-15.5}{1}{-12.5}
\arrowedsegment{-2}{-15.5}{-1}{-12.5}
\arrowedsegment{0}{16}{0}{13}
\arrowedsegment{2}{15.5}{1}{12.5}
\arrowedsegment{-2}{15.5}{-1}{12.5}
\arrowedsegment{17}{0}{14}{0}
\arrowedsegment{16.5}{1.5}{14}{0}
\arrowedsegment{16.5}{-1.5}{14}{0}
\arrowedsegment{-17}{0}{-14}{0}
\arrowedsegment{-16.5}{1.5}{-14}{0}
\arrowedsegment{-16.5}{-1.5}{-14}{0}
\arrowedsegment{9.5}{11}{6.5}{11}
\arrowedsegment{6.5}{14}{6.5}{11}
\arrowedsegment{8.5}{13}{6.5}{11}
\arrowedsegment{-9.5}{11}{-6.5}{11}
\arrowedsegment{-6.5}{14}{-6.5}{11}
\arrowedsegment{-8.5}{13}{-6.5}{11}
\arrowedsegment{9.5}{-11}{6.5}{-11}
\arrowedsegment{6.5}{-14}{6.5}{-11}
\arrowedsegment{8.5}{-13}{6.5}{-11}
\arrowedsegment{-9.5}{-11}{-6.5}{-11}
\arrowedsegment{-6.5}{-14}{-6.5}{-11}
\arrowedsegment{-8.5}{-13}{-6.5}{-11}
\draw[thick,fill=white] (0,0) circle [radius=3];
\draw[thick,fill=white] (10,6) circle [radius=2];
\draw[thick,fill=white] (10,-6) circle [radius=2];
\draw[thick,fill=white] (-10,6) circle [radius=2];
\draw[thick,fill=white] (-10,-6) circle [radius=2];
\draw[thick,fill=white] (0,12) circle [radius=2];
\draw[thick,fill=white] (0,-12) circle [radius=2];
\draw[thick,fill=white] (14,0) circle [radius=1];
\draw[thick,fill=white] (-14,0) circle [radius=1];
\draw[thick,fill=white] (6.5,11) circle [radius=1];
\draw[thick,fill=white] (-6.5,11) circle [radius=1];
\draw[thick,fill=white] (6.5,-11) circle [radius=1];
\draw[thick,fill=white] (-6.5,-11) circle [radius=1];
}
\draw[thick,fill=white] (0,0) circle [radius=3];
\node at (0,0) {a};
\draw[thick,fill=white] (10,6) circle [radius=2];
\node at (10,6) {b};
\draw[thick,fill=white] (10,-6) circle [radius=2];
\node at (10,-6) {c};
\draw[thick,fill=white] (-10,6) circle [radius=2];
\node at (-10,6) {f};
\draw[thick,fill=white] (-10,-6) circle [radius=2];
\node at (-10,-6) {e};
\draw[thick,fill=white] (0,12) circle [radius=2];
\node at (0,12) {g};
\draw[thick,fill=white] (0,-12) circle [radius=2];
\node at (0,-12) {d};

\end{tikzpicture}
\end{center}
\caption{}
\end{subfigure}
\hfill
\begin{subfigure}[b]{.45\textwidth}
\begin{center}
\begin{tikzpicture}[scale=0.187]
\bellbrasegment{0}{1}{0}{11}{0.3}
\bluebrasegment{0}{-1}{0}{-11}{0.3}
\bellbrasegment{2}{1}{9}{5.2}{0.3}
\bluebrasegment{2}{-1}{9}{-5.2}{0.3}
\bellbrasegment{-2}{1}{-9}{5.2}{0.3}
\bluebrasegment{-2}{-1}{-9}{-5.2}{0.3}
\draw[ultra thick,blue] (0,-6) to [out=0,in=-120] (5.5,-3.1);
\draw[ultra thick,blue] (0,-6) to [out=180,in=-60] (-5.5,-3.1);
\bellbrasegment{11}{5}{14}{0}{0.3}
\bellbrasegment{11}{-5}{14}{0}{0.3}
\bellbrasegment{-11}{5}{-14}{0}{0.3}
\bellbrasegment{-11}{-5}{-14}{0}{0.3}
\bellbrasegment{-9.5}{7}{-6.5}{11}{0.3}
\bellbrasegment{9.5}{7}{6.5}{11}{0.3}
\bellbrasegment{-9.5}{-7}{-6.5}{-11}{0.3}
\bellbrasegment{9.5}{-7}{6.5}{-11}{0.3}
\bellbrasegment{-1}{12}{-6.5}{11}{0.3}
\bellbrasegment{1}{12}{6.5}{11}{0.3}
\bellbrasegment{-1}{-12}{-6.5}{-11}{0.3}
\bellbrasegment{1}{-12}{6.5}{-11}{0.3}
\arrowedsegment{14}{6}{11.5}{6}
\polararrowedsegment{13}{7.5}{210}{2}
\arrowedsegment{12}{9.5}{10.5}{7.5}
\arrowedsegment{-14}{6}{-11.5}{6}
\polararrowedsegment{-13}{7.5}{-30}{2}
\arrowedsegment{-12}{9.5}{-10.5}{7.5}
\arrowedsegment{-14}{-6}{-11.5}{-6}
\polararrowedsegment{-13}{-7.5}{30}{2}
\arrowedsegment{-12}{-9.5}{-10.5}{-7.5}
\arrowedsegment{14}{-6}{11.5}{-6}
\polararrowedsegment{13}{-7.5}{150}{2}
\arrowedsegment{12}{-9.5}{10.5}{-7.5}
\arrowedsegment{0}{-16}{0}{-13}
\arrowedsegment{2}{-15.5}{1}{-12.5}
\arrowedsegment{-2}{-15.5}{-1}{-12.5}
\arrowedsegment{0}{16}{0}{13}
\arrowedsegment{2}{15.5}{1}{12.5}
\arrowedsegment{-2}{15.5}{-1}{12.5}
\arrowedsegment{17}{0}{14}{0}
\arrowedsegment{16.5}{1.5}{14}{0}
\arrowedsegment{16.5}{-1.5}{14}{0}
\arrowedsegment{-17}{0}{-14}{0}
\arrowedsegment{-16.5}{1.5}{-14}{0}
\arrowedsegment{-16.5}{-1.5}{-14}{0}
\arrowedsegment{9.5}{11}{6.5}{11}
\arrowedsegment{6.5}{14}{6.5}{11}
\arrowedsegment{8.5}{13}{6.5}{11}
\arrowedsegment{-9.5}{11}{-6.5}{11}
\arrowedsegment{-6.5}{14}{-6.5}{11}
\arrowedsegment{-8.5}{13}{-6.5}{11}
\arrowedsegment{9.5}{-11}{6.5}{-11}
\arrowedsegment{6.5}{-14}{6.5}{-11}
\arrowedsegment{8.5}{-13}{6.5}{-11}
\arrowedsegment{-9.5}{-11}{-6.5}{-11}
\arrowedsegment{-6.5}{-14}{-6.5}{-11}
\arrowedsegment{-8.5}{-13}{-6.5}{-11}
\draw[thick,fill=white] (0,0) circle [radius=3];
\node at (0,0) {a};
\draw[thick,fill=white] (10,6) circle [radius=2];
\node at (10,6) {b};
\draw[thick,fill=white] (10,-6) circle [radius=2];
\node at (10,-6) {c};
\draw[thick,fill=white] (-10,6) circle [radius=2];
\node at (-10,6) {f};
\draw[thick,fill=white] (-10,-6) circle [radius=2];
\node at (-10,-6) {e};
\draw[thick,fill=white] (0,12) circle [radius=2];
\node at (0,12) {g};
\draw[thick,fill=white] (0,-12) circle [radius=2];
\node at (0,-12) {d};
\draw[thick,fill=white] (14,0) circle [radius=1];
\draw[thick,fill=white] (-14,0) circle [radius=1];
\draw[thick,fill=white] (6.5,11) circle [radius=1];
\draw[thick,fill=white] (-6.5,11) circle [radius=1];
\draw[thick,fill=white] (6.5,-11) circle [radius=1];
\draw[thick,fill=white] (-6.5,-11) circle [radius=1];
\end{tikzpicture}
\end{center}
\caption{}
\end{subfigure}
\caption{The four networks included in the set $F$ considered in text. The projecting states shown in blue in (a) are defined by an operator $\mathcal{O}$ acting on the maximally entangled state, these can be pushed to any subset of three legs, as seen in (b-d). In general the pushed through operator will not have a tensor product structure, so the corresponding projecting state will be entangled across three legs. This is indicated here by the thick blue line in (b-d). All four networks shown here have the same boundary state. The entropies of subsets of boundary legs is calculated by choosing the network which contains an isometric cut enclosing those boundary legs, and applying the isometric cut formula. No one network contains an isometric cut for every boundary region, but the set of four networks together do.}
\label{fig:dynamicexample}
\end{center}
\end{figure}
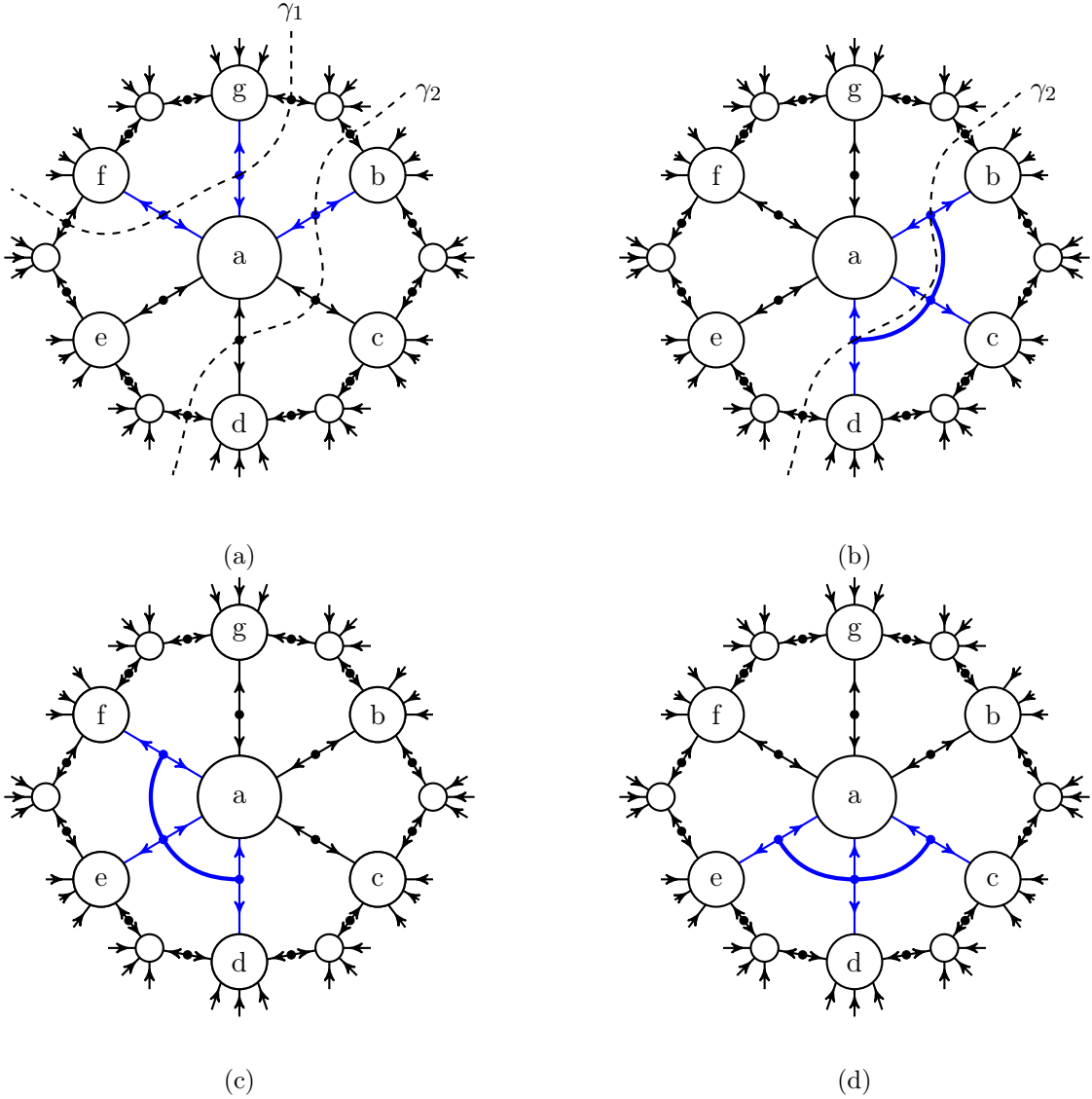

\section{Dynamic tensor network states}\label{sec:dynamicexample}

Consider an AdS spacetime and a spacelike slice of its boundary. The maximin formula shows that the entanglement entropy of a given boundary region can be found by determining the area of an extremal surface extending into the AdS spacetime. For a dynamic spacetime these extremal surfaces may lie in many different slices of the interior. In the tensor network picture we have identified isometric cuts as the network analogue of extremal surfaces. Further, we have suggested a set of networks contracting to a single boundary state is the analogue of the set of spacelike slices of the bulk spacetime. A tensor network state which is analogous to an evolving spacetime then should have isometric cuts for different boundary regions living in different networks drawn from this set. 

We will call this set of networks $F$, and specify the networks it contains by giving a defining network $N_0$ along with a set of allowed transformation rules. Continuing our analogy, we view these transformations as corresponding to deformations of the interior spacelike slices. Importantly, these transformations must preserve the boundary state. An example of such a transformation was given as equation \ref{eq:slickpsifree}.

To construct examples of boundary states with a geometry corresponding to a dynamic spacetime we begin with the example network of figure \ref{fig:sixlegexample}a and replace three of the maximally entangled pairs which are projected into the central vertex with the state
\begin{align}
\ket{\Psi_i}=(\mathcal{O} \otimes \mathcal{I}) \ket{\Psi^+}.
\end{align}
This is our defining network, shown in figure \ref{fig:dynamicexample}a. The allowed transformations we take to be the operator pushing operation discussed in section \ref{sec:basics}. This results in the four networks shown in figure \ref{fig:dynamicexample}b-d being included in the set $F$.

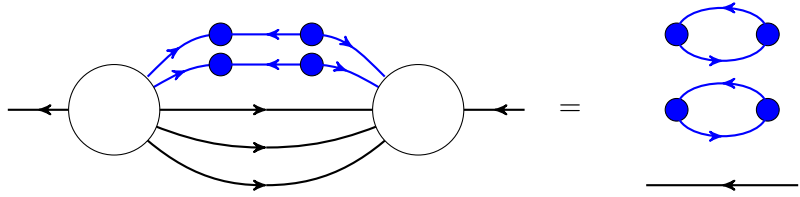
\begin{figure}
\begin{center}
\begin{tikzpicture}[scale=0.2]
\draw[thick,postaction={on each segment={mid arrow}}] (-3,0) -- (-7,0);
\draw[thick,postaction={on each segment={mid arrow}}] (27,0) -- (23,0);

\draw[thick,blue,postaction={on each segment={mid arrow}}] (2.59,1.5) to [out=30,in=180] (7,3);
\draw[fill=blue] (7,3) circle [radius=0.75];
\draw[thick,blue,postaction={on each segment={mid arrow}}] (13,3)--(7,3);
\draw[fill=blue] (13,3) circle [radius=0.75];
\draw[thick,blue,postaction={on each segment={mid arrow}}] (13,3)to [out=0,in=150] (17.4,1.5);

\draw[thick,->] (0,0) to [out=-22.5,in=180] (10,-2.5);
\draw[thick,->] (10,-2.5) to [out=0,in=-157.5] (20,0);
\draw[thick,->] (0,0) to [out=0,in=180] (10,0);
\draw[thick,->] (10,0) to [out=0,in=180] (20,0);

\draw[thick,blue,postaction={on each segment={mid arrow}}] (2.2,2.2) to [out=45,in=180] (7,5);
\draw[fill=blue] (7,5) circle [radius=0.75];
\draw[thick,blue,postaction={on each segment={mid arrow}}] (13,5)--(7,5);
\draw[fill=blue] (13,5) circle [radius=0.75];
\draw[thick,blue,postaction={on each segment={mid arrow}}] (13,5)to [out=0,in=135] (17.8,2.2);
\draw[thick,->] (0,0) to [out=-45,in=180] (10,-5);
\draw[thick,->] (10,-5) to [out=0,in=-135] (20,0);
\draw[fill=white] (0,0) circle [radius=3];
\draw[fill=white] (20,0) circle [radius=3];
\draw node at (30,0) {$=$};
\draw[thick,postaction={on each segment={mid arrow}}] (45,-5) -- (35,-5);

\draw[fill=blue] (37,5) circle [radius=0.75];
\draw[fill=blue] (43,5) circle [radius=0.75];
\draw[thick,blue,postaction={on each segment={mid arrow}}] (43,5) to [out=90,in=90] (37,5);
\draw[thick,blue,postaction={on each segment={mid arrow}}] (37,5) to [out=-90,in=-90] (43,5);

\draw[fill=blue] (37,0) circle [radius=0.75];
\draw[fill=blue] (43,0) circle [radius=0.75];
\draw[thick,blue,postaction={on each segment={mid arrow}}] (43,0) to [out=90,in=90] (37,0);
\draw[thick,blue,postaction={on each segment={mid arrow}}] (37,0) to [out=-90,in=-90] (43,0);

\end{tikzpicture}
\end{center}
\caption{The identity used to show a cut containing one interior leg which is blue in \ref{fig:dynamicexample}a is isometric up to a normalization factor. This identity is easily derived from that in figure \ref{fig:subperfect}.}
\label{fig:subperfect2}
\end{figure}

\begin{figure}
\begin{center}
\begin{tikzpicture}[scale=0.2]
\draw[thick,postaction={on each segment={mid arrow}}] (-2.2,-2.2) -- (-5,-5);
\draw[thick,postaction={on each segment={mid arrow}}] (-2.2,2.2) -- (-5,5);
\draw[thick,postaction={on each segment={mid arrow}}] (25,-5)--(22.2,-2.2);
\draw[thick,postaction={on each segment={mid arrow}}] (25,5) -- (22.2,2.2);
\draw[thick,->] (0,0) to [out=22.5,in=180] (10,2.5);
\draw[thick,->] (10,2.5) to [out=0,in=157.5] (20,0);
\draw[thick,->] (0,0) to [out=-22.5,in=180] (10,-2.5);
\draw[thick,->] (10,-2.5) to [out=0,in=-157.5] (20,0);
\draw[thick,blue,postaction={on each segment={mid arrow}}] (2.2,2.2) to [out=45,in=180] (7,5);
\draw[fill=blue] (7,5) circle [radius=0.75];
\draw[thick,blue,postaction={on each segment={mid arrow}}] (13,5)--(7,5);
\draw[fill=blue] (13,5) circle [radius=0.75];
\draw[thick,blue,postaction={on each segment={mid arrow}}] (13,5)to [out=0,in=135] (17.8,2.2);
\draw[thick,->] (0,0) to [out=-45,in=180] (10,-5);
\draw[thick,->] (10,-5) to [out=0,in=-135] (20,0);
\draw[fill=white] (0,0) circle [radius=3];
\draw[fill=white] (20,0) circle [radius=3];
\draw node at (30,0) {$=$};
\draw[thick,postaction={on each segment={mid arrow}}] (45,-5) -- (35,-5);
\draw[thick,postaction={on each segment={mid arrow}}] (45,0) -- (35,0);
\draw[fill=blue] (37,5) circle [radius=0.75];
\draw[fill=blue] (43,5) circle [radius=0.75];
\draw[thick,blue,postaction={on each segment={mid arrow}}] (43,5) to [out=90,in=90] (37,5);
\draw[thick,blue,postaction={on each segment={mid arrow}}] (37,5) to [out=-90,in=-90] (43,5);
\end{tikzpicture}
\end{center}
\caption{The identity used to show the cut $\gamma_1$ in figure \ref{fig:dynamicexample}a is isometric up to a normalization factor.}
\label{fig:subperfect}
\end{figure}
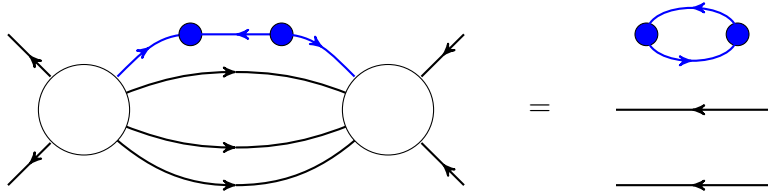

It is not difficult to see that no one of these networks contains isometric cuts for every boundary region, but the set of four together do. Consider for example the defining network, figure \ref{fig:dynamicexample}a, with $\mathcal{O}$ sitting on three of the interior legs (marked as blue legs). Consider first the subregion consisting of exterior legs attached to vertex f. The cut enclosing f and crossing three legs is isometric, this follows from the identity shown in figure \ref{fig:subperfect2}. Similarly, the cut $\gamma_2$ which encloses f and g can be shown to be isometric by use of the identity in figure \ref{fig:subperfect}. Further, the boundary legs adjacent to vertices f,g,b or c,d,e are also enclosed by isometric cuts contained in the network of figure \ref{fig:dynamicexample}a, since a minimal cut which crosses three interior blue legs is isometric. 

The need for additional networks to be included in the set $F$ arises when we consider subregions adjacent to two or fewer black interior legs. Take for example the subregion containing vertices b, c and d. The cut $\gamma_2$ which crosses one blue and two black interior legs is not isometric, nor is the other possible cut which crosses one black and two blue interior legs. To find an isometric cut enclosing b-c-d we consider network \ref{fig:dynamicexample}b. Since all the operator insertions now live on the cut $\gamma_2$ and $\gamma_2$ crosses a minimal number of legs, it is isometric. Similarly, an isometric cut for enclosing d-e-f can be found in network \ref{fig:dynamicexample} and an isometric cut for e-f-g in network \ref{fig:dynamicexample}. 

A complication arises in considering the region containing only c, d or e, or regions f-e, e-d, d-c, c-b. Consider region d-c, the remaining possibilities are handled similarly. In this case an isometric cut can be found in network \ref{fig:dynamicexample}b. To see this, we must return to the notion of a cut through the network. Recall that a cut $\gamma$ corresponds to a specification of projecting state, on which operators $C$ and $D$ act to prepare the boundary state. Implicitly, choosing a cut involves specifying which legs of the projecting state are acted on by the operator $C$ and which by the operator $D$, corresponding to our breakdown of the projecting state Hilbert space into $\mathcal{H}_B$ and $\mathcal{H}_{\bar{B}}$. To specify this in the graphical notation we can use a double line to specify a cut through the network. One cut $\gamma$ crosses the legs which are associated with the $B$ Hilbert space, with a second cut $\bar{\gamma}$ denoting the legs in the $\bar{B}$ Hilbert space. We have adopted this notation in figure \ref{fig:doubleline} to specify an isometric cut for the region c-d. It is straightforward to show that the operators defined by this cut are isometric, from which we can conclude that the mutual information across the $B$ and $\bar{B}$ subsystems of the projecting state is equal to the entropy of the boundary region c-d, as needed. 

\begin{figure}
\begin{center}
\begin{tikzpicture}[scale=0.187]
\bellbrasegment{0}{1}{0}{11}{0.3}
\bluebrasegment{0}{-1}{0}{-11}{0.3}
\bluebrasegment{2}{1}{9}{5.2}{0.3}
\bluebrasegment{2}{-1}{9}{-5.2}{0.3}
\bellbrasegment{-2}{1}{-9}{5.2}{0.3}
\bellbrasegment{-2}{-1}{-9}{-5.2}{0.3}
\draw[ultra thick,blue] (0,-6) to [out=0,in=-120] (5.5,-3.1);
\draw[ultra thick,blue] (5.5,-3.1) to [out=60,in=-60] (5.5,3.1);


\draw[thick, dashed] (10,2) to [out=10,in=210] (17,7);
\draw[thick, dashed] (5,5) to [out=30,in=180] (10,2);
\draw[thick, dashed] (5,5) to [out=210,in=60] (4,-2.6);
\draw[thick, dashed] (4,-2.6) to [out=240,in=30] (0,-4.5);
\draw[thick, dashed] (0,-4.5) to [out=210,in=80] (-4.75,-11.5);
\draw[thick, dashed] (-4.75,-11.5) to [out=-100,in=-100] (-5.5,-15.5);

\draw[thick, dashed] (-2.75,-11.5) to [out=-100,in=80] (-3.5,-16);
\draw[thick, dashed] (0,-8) to [out=210,in=80] (-2.75,-11.5);

\draw[thick, dashed] (0,-8) to [out=20,in=-120] (7,-3.6);
\draw[thick, dashed] (7,-3.6) to [out=70,in=-160] (10,0);
\draw[thick, dashed] (10,0) to [out=15,in=-160] (17,5);

\node[below] at (-3.5,-16) {$\gamma$};
\node[below] at (-5.5,-15.5) {$\bar{\gamma}$};

\bellbrasegment{11}{5}{14}{0}{0.3}
\bellbrasegment{11}{-5}{14}{0}{0.3}
\bellbrasegment{-11}{5}{-14}{0}{0.3}
\bellbrasegment{-11}{-5}{-14}{0}{0.3}
\bellbrasegment{-9.5}{7}{-6.5}{11}{0.3}
\bellbrasegment{9.5}{7}{6.5}{11}{0.3}
\bellbrasegment{-9.5}{-7}{-6.5}{-11}{0.3}
\bellbrasegment{9.5}{-7}{6.5}{-11}{0.3}
\bellbrasegment{-1}{12}{-6.5}{11}{0.3}
\bellbrasegment{1}{12}{6.5}{11}{0.3}
\bellbrasegment{-1}{-12}{-6.5}{-11}{0.3}
\bellbrasegment{1}{-12}{6.5}{-11}{0.3}
\arrowedsegment{14}{6}{11.5}{6}
\polararrowedsegment{13}{7.5}{210}{2}
\arrowedsegment{12}{9.5}{10.5}{7.5}
\arrowedsegment{-14}{6}{-11.5}{6}
\polararrowedsegment{-13}{7.5}{-30}{2}
\arrowedsegment{-12}{9.5}{-10.5}{7.5}
\arrowedsegment{-14}{-6}{-11.5}{-6}
\polararrowedsegment{-13}{-7.5}{30}{2}
\arrowedsegment{-12}{-9.5}{-10.5}{-7.5}
\arrowedsegment{14}{-6}{11.5}{-6}
\polararrowedsegment{13}{-7.5}{150}{2}
\arrowedsegment{12}{-9.5}{10.5}{-7.5}
\arrowedsegment{0}{-16}{0}{-13}
\arrowedsegment{2}{-15.5}{1}{-12.5}
\arrowedsegment{-2}{-15.5}{-1}{-12.5}
\arrowedsegment{0}{16}{0}{13}
\arrowedsegment{2}{15.5}{1}{12.5}
\arrowedsegment{-2}{15.5}{-1}{12.5}
\arrowedsegment{17}{0}{14}{0}
\arrowedsegment{16.5}{1.5}{14}{0}
\arrowedsegment{16.5}{-1.5}{14}{0}
\arrowedsegment{-17}{0}{-14}{0}
\arrowedsegment{-16.5}{1.5}{-14}{0}
\arrowedsegment{-16.5}{-1.5}{-14}{0}
\arrowedsegment{9.5}{11}{6.5}{11}
\arrowedsegment{6.5}{14}{6.5}{11}
\arrowedsegment{8.5}{13}{6.5}{11}
\arrowedsegment{-9.5}{11}{-6.5}{11}
\arrowedsegment{-6.5}{14}{-6.5}{11}
\arrowedsegment{-8.5}{13}{-6.5}{11}
\arrowedsegment{9.5}{-11}{6.5}{-11}
\arrowedsegment{6.5}{-14}{6.5}{-11}
\arrowedsegment{8.5}{-13}{6.5}{-11}
\arrowedsegment{-9.5}{-11}{-6.5}{-11}
\arrowedsegment{-6.5}{-14}{-6.5}{-11}
\arrowedsegment{-8.5}{-13}{-6.5}{-11}
\draw[thick,fill=white] (0,0) circle [radius=3];
\node at (0,0) {a};
\draw[thick,fill=white] (10,6) circle [radius=2];
\node at (10,6) {b};
\draw[thick,fill=white] (10,-6) circle [radius=2];
\node at (10,-6) {c};
\draw[thick,fill=white] (-10,6) circle [radius=2];
\node at (-10,6) {f};
\draw[thick,fill=white] (-10,-6) circle [radius=2];
\node at (-10,-6) {e};
\draw[thick,fill=white] (0,12) circle [radius=2];
\node at (0,12) {g};
\draw[thick,fill=white] (0,-12) circle [radius=2];
\node at (0,-12) {d};
\draw[thick,fill=white] (14,0) circle [radius=1];
\draw[thick,fill=white] (-14,0) circle [radius=1];
\draw[thick,fill=white] (6.5,11) circle [radius=1];
\draw[thick,fill=white] (-6.5,11) circle [radius=1];
\draw[thick,fill=white] (6.5,-11) circle [radius=1];
\draw[thick,fill=white] (-6.5,-11) circle [radius=1];
\end{tikzpicture}
\end{center}
\caption{Illustration of how to choose an isometric cut for the subregion consisting of legs adjacent to the d and c vertices. The upper dashed line crosses legs included in the $\bar{B}$ Hilbert space, while the lower dashed line crosses legs included in the $B$ Hilbert space. It is straightforward to check that both the operators above and below the dashed lines are isometries; it follows that $S(A) = \frac{1}{2}I(\bar{B}:B)$.}
\label{fig:doubleline}
\end{figure}
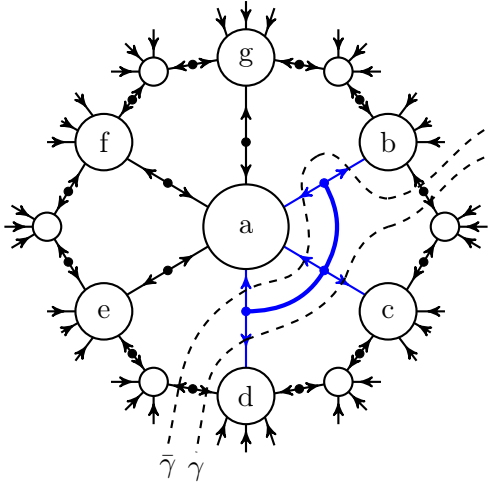

One advantage to the isometric cut formula is that it is not necessary to limit the networks which are included in the set $F$ searched over. Indeed, it is a logical consequence of our definitions that any cut which is isometric will have as its length the entropy of the enclosed boundary region. Thus although we specified only the four networks given in figure \ref{fig:dynamicexample} we could add arbitrary networks to the set $F$. Those without isometric cuts would not disturb the isometric cut formula at all, and any additional networks having isometric cuts would give unchanged values for the boundary entropies. 

It is not difficult to construct further examples of sets of networks satisfying the isometric cut formula based on the construction used here. Indeed, we may continue the pattern of contraction given by figure \ref{fig:sixlegexample}a and construct networks with an arbitrary number of layers. We can then proceed to replace a subset of the projecting maximally entangled pairs by non-maximally entangled states; allowing the same freedom of pushing operators through adjacent tensors then gives a set of networks associated with a fixed boundary state. We have not yet systematically studied which networks defined in this way contain isometric cuts for all possible boundary regions; doing so remains a direction for future work.

\section{Discussion and open questions} \label{sec:discussion}

Tensor network models realize key features of the AdS/CFT correspondence, including the connection between entanglement and geometry given by the Ryu-Takayanagi formula. Here we have extended tensor network models to include a larger class of states with entropies calculated using cuts in a set of networks. To achieve this, we have restated the Ryu-Takayanagi formula for networks as the isometric cut formula,
\begin{align}
S(A) = L(\gamma_{iso}^A)
\end{align}
where $\gamma_{iso}^A$ is an isometric cut, which means the map defined by the network on either side of $\gamma$ is an isometry from the cut legs to the boundary. This isometric cut formula for networks is analogous to the maximin formula in AdS/CFT, with isometric cuts in networks playing the role of extremal surfaces in AdS/CFT. This identification of isometric cuts as the network analogue of extremal surfaces is supported by the reduction of isometric cuts to minimal surfaces in the case of a HaPPY network. Additionally, this identification leads to an assignment of length in the network which has the expected geometric properties of additivity and gauge invariance. 

It would be interesting to modify the definition of length given here, which as stated relies on the existence of an isometric cut crossing the leg of interest, to a purely local definition. Such a definition would allow the geometric procedure for determining boundary entropies to be stated as a maximin formula rather than the reformulation as the isometric cut formula used here. Additionally one would know all of the lengths in every interior slice, something not possible with the isometric cut definition. One natural possibility is to define the length of a leg by the mutual information of the corresponding projecting state. However recovering the maximin formula using such a definition is problematic. This is because we can write the projecting state as $\ket{\Psi_i} = \mathcal{O} \otimes \mathcal{I} \ket{\Psi^+}$ and move $\mathcal{O}$ off of the projecting state, reconsidering it as part of a nearby tensor. In such a network the length of the leg $\gamma_i$ considered will be $\log \dim \gamma_i$. If networks with the operator $\mathcal{O}$ moved off of all the projecting states are included in the set optimized over the maximin formula will always return cuts with $L(\gamma) = \log \dim \gamma$, which restricts the models to describing static spacetimes. Thus one would need appropriate limitations on the set of networks optimized over if such a definition of length were used. 

One motivation for this work is to try and learn something about how a tensor network model can be evolved in spacetime. As a first step, we have clearly identified the HaPPY networks as representing static states. This indicates HaPPY states are eigenstates of an appropriate Hamiltonian, although it is not clear how to construct this Hamiltonian or which set of HaPPY states correspond to a single Hamiltonian. Also in this context, the variations of the interior of the network which are needed to find isometric cuts for all boundary regions are most easily interpreted as evolving a portion of the network (forward or backward) in time. By studying which such operations are required so that all isometric cuts are included we may learn something about how to apply time evolution operators to a network.

The perspective here begins by specifying a network. This defines a boundary state and associated geometry. In appendix \ref{sec:networksfromgeometry} we argue that it is also possible to begin with a geometry and construct a network and boundary state. It is interesting to consider proceeding in another direction - beginning with a quantum state, when can we construct a network which has an appropriate geometry and contracts to the given boundary state? This question is the network analogue of asking which CFT states have gravity duals. Our contribution indicates that this question should be modified somewhat, to ask which quantum states have an associated set of networks which satisfy the isometric cut formula and all contract to the given boundary state. 

\section{Acknowledgements}

This work was carried out under the supervision of Mark Van Raamsdonk, who was involved in discussions throughout this project and reviewed drafts of this manuscript. Charles Rabideau made useful contributions in the early stages of the project. We are also indebted to Michael Walter, Grant Salton, Zhao Yang, and David Stephen for helpful discussions. Dominik Neuenfeld and Jaehoon Lee provided feedback on early versions of the manuscript.

AM was partially supported by the It from Qubit Collaboration, which is sponsored by the Simons Foundation. AM was also supported by a CGSM award given by the National Research Council of Canada. This research benefited from the It from Qubit summer school held at the Perimeter Institute for Theoretical Physics. Research at Perimeter Institute is supported by the Government of Canada through Industry Canada and by the Province of Ontario through the Ministry of Economic Development \& Innovation.

\bibliographystyle{unsrt}
\bibliography{biblio}

\appendix

\section{Networks from geometry} \label{sec:networksfromgeometry}

Typically in the HaPPY or random tensor network construction one begins with a network. The network defines a boundary state and a geometry. In the random tensor construction it is possible to begin with a geometry and construct a network which has minimal surfaces matching those of the given geometry. This was already claimed in ref. \cite{hayden2016holographic}, however, we argue that the construction given there is problematic and give an alternative construction. We also wish to acknowledge that our construction borrows a technique from ref. \cite{bao2015holographic}. 

We consider a disk $M = \{ (x,y) : 0 \leq x^2+y^2 \leq 1\}$ which is endowed with a distance function $d(u,v)$. Our goal is to fill in the disk with a planar tensor network which satisfies the Ryu-Takayanagi formula and whose minimal surfaces have lengths approximating the function $d(u,v)$. We introduce a parameter $\delta$ which represents a unit of length in the continuous geometry. In particular we say the network approximates the geometry of the disk to a resolution of $\epsilon$ if
\begin{align} \label{eq:approx}
\left|d(u,v) - \frac{\delta}{\log D} \cdot L(u,v)\right| \leq \epsilon,
\end{align}
where $L(u,v)$ is the graph length in the network, $D$ is the dimension of the legs in the network. The $\delta / \log{D}$ should be understood as a conversion factor from graph length (unitless) to physical length.

\begin{figure}
\begin{center}
\begin{tikzpicture}[scale=0.2]
\draw (0,0) circle [radius=15]; 
\draw[blue,thick] (10.6,10.6) -- (10.6,-10.6) -- (-10.6,-10.6) --(-10.6,10.6) --(10.6,10.6);
\draw[blue,thick] (10.6,10.6) -- (-10.6,-10.6);
\draw[blue,thick] (10.6,-10.6) -- (-10.6,10.6);
\node[right] at (0,0) {$O$};
\node[right] at (10.6,10.6) {$A$};
\node[right] at (10.6,-10.6) {$B$};
\node[left] at (-10.6,-10.6) {$C$};
\node[left] at (-10.6,10.6) {$D$};
\end{tikzpicture}
\end{center}
\caption{Illustration of our procedure for constructing a network whose minimal lengths approximate those of a given geometry. The example shown constructs a network with four boundary legs which approximates a disk shaped region of $\mathbb{R}^2$. In the first step, four boundary points are chosen and all of the minimal cuts anchored on those points are drawn. The minimal cuts form a planar graph, in this example the graph has vertices $A,B,C,D$ and $O$.}
\label{fig:stepone}
\end{figure}
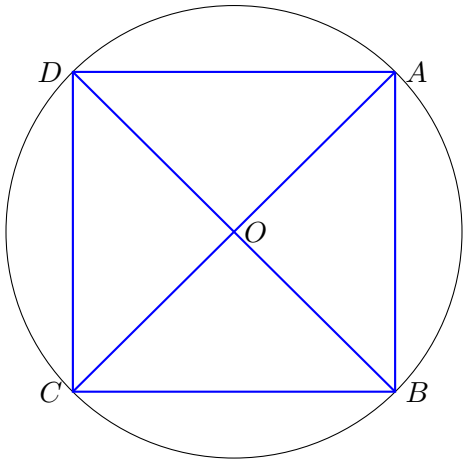

Our strategy is to construct a graph whose minimal cuts satisfy \ref{eq:approx} and then populate that graph with random tensors of large bond dimension. With random tensors placed on the vertices, the results of ref. \cite{hayden2016holographic} then guarantee the Ryu-Takayanagi formula is satisfied. Let us consider as an example a disk which is a section of flat space, so $d(u,v) = \sqrt{(u_1-v_1)^2+(u_2-v_2)^2}$. A reasonable first approach is to tile the disk with a regular polygon. This can be done with triangles, squares, or hexagons. However, none of these tilings correctly reproduce lengths in the disk in the sense of \ref{eq:approx}. For example in a tiling with squares, the graph length function is 
\begin{align}
L(u,v) = |u_1 - v_1| + |u_2 - v_2|,
\end{align}
which doesn't approximate the Euclidean distance. Additionally, such a distance function gives highly degenerate minimal surfaces - for example a staircase shaped path gives the same distance between two boundary points as a path which turns only once. Regular tilings using triangles or hexagons produce similar graph distance functions and also have degenerate minimal surfaces.

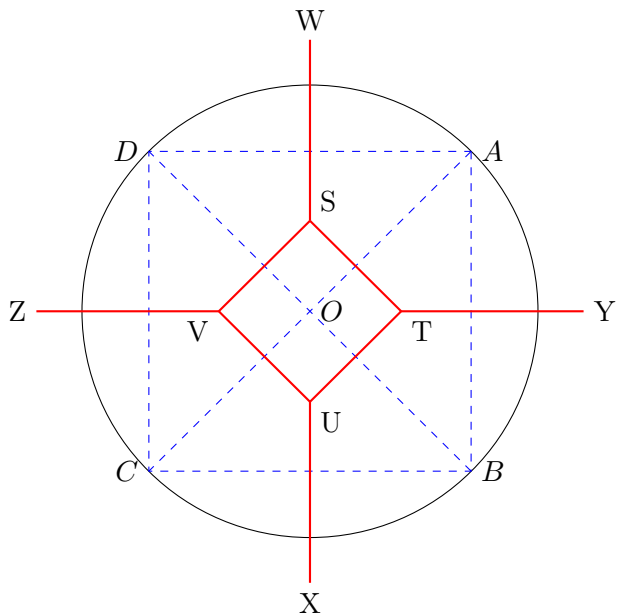
\begin{figure}
\begin{center}
\begin{tikzpicture}[scale=0.2]
\draw (0,0) circle [radius=15]; 
\draw[dashed, blue, thin] (10.6,10.6) -- (10.6,-10.6) -- (-10.6,-10.6) --(-10.6,10.6) --(10.6,10.6);
\draw[dashed,blue,thin] (10.6,10.6) -- (-10.6,-10.6);
\draw[dashed,blue,thin] (10.6,-10.6) -- (-10.6,10.6);
\node[right] at (0,0) {$O$};
\node[right] at (10.6,10.6) {$A$};
\node[right] at (10.6,-10.6) {$B$};
\node[left] at (-10.6,-10.6) {$C$};
\node[left] at (-10.6,10.6) {$D$};

\draw[red,thick] (6,0) -- (18,0);
\draw[red,thick] (-6,0) -- (-18,0);
\draw[red,thick] (0,6) -- (0,18);
\draw[red,thick] (0,-6) -- (0,-18);
\draw[red,thick] (0,6) -- (6,0) -- (0,-6) -- (-6,0) -- (0,6);

\node[above right] at (0,6) {S};
\node[below right] at (6,0) {T};
\node[below right] at (0,-6) {U};
\node[below left] at (-6,0) {V};

\node[above] at (0,18) {W};
\node[below] at (0,-18) {X};
\node[right] at (18,0) {Y};
\node[left] at (-18,0) {Z};

\end{tikzpicture}
\end{center}
\caption{Illustration of the second step in our procedure for constructing a network which approximates a given geometry. In this step, the dual of the graph formed in step one is drawn. The edges of the dual graph are assigned a weight based on the $\mathbb{R}^2$ length of the edges in the direct graph they cut. For example, a weight of $\text{Floor}( \bar{AB} / \delta)$ is assigned to the edge $TY$, where $\delta x$ is a parameter with units of length controlling how closely the graph approximates lengths in the disk.}
\label{fig:steptwo}
\end{figure}

Our construction begins by specifying a set of points on the edge of the disk. The closeness of our approximation is set in part by the number of points on the boundary chosen, which we will denote by $N$. The construction proceeds by drawing every minimal surface between pairs of these points; this is illustrated in figure \ref{fig:stepone}. The resulting surfaces define a graph which we take to be the dual graph of the tensor network being constructed. Importantly, the edges in the direct graph are assigned a weighting $w_i$ set by
\begin{align}
w_i = \text{Floor}(d(u_j,u_k) / \delta).
\end{align}
Up to the rounding implemented by the floor function, the weight of the edges in the direct graph is given by the length of the edges in the dual graph which they cut, measured in units of $\delta$. Forming the direct graph from the dual graph and assigning the weightings is illustrated in figure \ref{fig:steptwo}.  

Finally, random tensors are placed on the vertices of the direct graph and the number of legs along an edge is chosen to be equal to the weighting $w_i$ associated with that leg. We can then show that the resulting tensor network has lengths which satisfy \ref{eq:approx}. To prove this, note that a cut in the tensor network is also a path in the dual graph. We will consider a minimal cut passing from $u_0\rightarrow u_N$ where the $u_i$ are vertices in the dual graph. Consider one segment of that path which passes from $u_i$ to $u_j$. The length of this segment is given by 
\begin{align}
L(u_i,u_j) = (\text{number of legs crossed})\log D.
\end{align}
The number of legs crossed is just $w_i$, which is the length of that segment of the path given in units of $\delta$, 
\begin{align}
L(u_i,u_j) = \text{Floor}\left(\frac{d(u_i,u_j)}{\delta}\right) \cdot \log D
\end{align}
Inserting this into \ref{eq:approx} gives that
\begin{align} 
\left|d(u_i,u_j) - \frac{\delta}{\log D} \cdot L(u_i,u_j)\right| \leq \delta.
\end{align}
The number of boundary points chosen, $N$, sets the maximal number of segments in a minimal cut through the disk, which we call $f(N)$\footnote{A chord $AB$ divides the points $C,D...$ on the circles edge into two sets of size $n_1$ and $n_2$ where $n_1+n_2\leq N-2$. Since every pairing of such points gives a chord which crosses $AB$ once we can bound the number of cuts through $AB$ by $((N-2)/2)^2$}. Then the triangle inequality gives that for a minimal path through the network
\begin{align}
d(u_0,u_n) \leq \delta f(N).
\end{align}
A network with resolution $\epsilon$ then can be constructed by choosing $\delta = \epsilon / f(N)$.

\end{document}